\documentclass[journal, onecolumn]{IEEEtran}
\usepackage[utf8]{inputenc} 
\usepackage[T1]{fontenc}
\usepackage{url}
\usepackage{ifthen}
\usepackage{cite}
\usepackage{amssymb, amsfonts}
\usepackage[cmex10]{amsmath} 
\usepackage{booktabs}
\usepackage{xcolor}
\usepackage{cases}
\usepackage[ruled,lined]{algorithm2e}
\usepackage{hyperref}
\usepackage{cleveref}
\usepackage{amsthm}                        
\usepackage{enumitem}
\usepackage{tikz}
\usepackage{thmtools}
\theoremstyle{plain}
\newtheorem{theorem}{Theorem}
\newtheorem{lemma}[theorem]{Lemma}
\newtheorem{proposition}[theorem]{Proposition}
\newtheorem{definition}[theorem]{Definition}

\newtheorem{conjecture}[theorem]{Conjecture}

\DeclareMathOperator*{\argmax}{arg\,max}

\theoremstyle{definition}
\newtheorem{remark}{Remark}

\newcommand{\lele}{\textcolor{black}}

\newcommand{\mic}{\textcolor{black}}

\newcommand\blfootnote[1]{%
  \begingroup
  \renewcommand\thefootnote{}\footnote{#1}%
  \addtocounter{footnote}{-1}%
  \endgroup
}

\ifCLASSINFOpdf
\else
\fi
\begin{document}
%
\title{Information-Theoretic Limits on \\
        Exact Subgraph Alignment Problem}
%
%
%

\author{Chun Hei Michael Shiu\blfootnote{Chun Hei Michael Shiu is with the Department of Electrical and Computer Engineering, University of British Columbia, Vancouver, BC V6T1Z4, Canada (email: shiuchm@student.ubc.ca).}, Hei Victor Cheng\blfootnote{Hei Victor Cheng is with the Department of Electrical and Computer Engineering, Aarhus University, Denmark (email: hvc@ece.au.dk).}, and Lele Wang\blfootnote{Lele Wang is with the Department of Electrical and Computer Engineering, University of British Columbia, Vancouver, BC V6T1Z4, Canada (email: lelewang@ece.ubc.ca).}  
}

\maketitle

\begin{abstract}
The graph alignment problem aims to identify the vertex correspondence between two correlated graphs. Most existing studies focus on the scenario in which the two graphs share the same vertex set. However, in many real-world applications, such as computer vision, social network analysis, and bioinformatics, the task often involves locating a \emph{small} graph pattern within a \emph{larger} graph. Existing graph alignment algorithms and analysis cannot directly address these scenarios because they are not designed to identify the specific subset of vertices where the small graph pattern resides within the larger graph.
Motivated by this limitation, we introduce the \emph{subgraph alignment problem}, which seeks to recover both the vertex set and/or the vertex correspondence of a small graph pattern embedded in a larger graph. In the special case where the small graph pattern is an induced subgraph of the larger graph and both the vertex set and correspondence are to be recovered, the problem reduces to the \emph{subgraph isomorphism problem}, which is NP-complete in the worst case. In this paper, we formally formulate the subgraph alignment problem by proposing the Erd\"os-R\'enyi subgraph pair model together with some appropriate recovery criterion. We then establish almost-tight information-theoretic results for the subgraph alignment problem and present some novel approaches for the analysis.
\end{abstract}


%
\IEEEpeerreviewmaketitle

\section{Introduction}

\blfootnote{This work was presented in part at the 2025 IEEE International Symposium on Information Theory \cite{shiu25} and submitted in part to the 2026 IEEE International Symposium on Information Theory.}

\IEEEPARstart{T}{he} graph alignment problem, also known as the graph matching problem or noisy graph isomorphism problem, has received increasing attention over the last decade due to its wide applications across multiple fields such as natural language processing~\cite{hag05,pru15}, bioinformatics~\cite{sin08,kou13,rau19}, and social network de-anonymization~\cite{nar09,kor14}, among others. The graph alignment problem aims to find the best vertex correspondence between two correlated graphs with respect to a certain criterion.
Extensive studies have been done to establish the information-theoretic limits~\cite{ped11, cul16, cul17, wu22} and efficient algorithms~\cite{dai19,fan20,din21,mao21,din23,mao23a,mao23b} for the graph alignment problem. Variations such as seeded alignment~\cite{kor14,yar13,mos20,shi17}, database alignment~\cite{cul18}, and alignment with additional attribute information~\cite{zha24,wan23,wan24} have also been studied.

Most existing work on graph alignment focuses on scenarios where the two graphs have the same number of vertices. However, in many real-world applications, such as computer vision \cite{lan23,sho00,lla01}, social network analysis \cite{fan2012,das20,ma18}, and bioinformatics \cite{borgelt2002,tia07,zha09}, the problem often involves locating a small graph pattern within a larger graph. 
For example, in chemistry and biology, when searching for a  specific (properties indicating) molecules in a large structure, the goal is to identify and locate the set of vertices of the subgraph representing the molecules \cite{borgelt2002}.
In applications such as social network analysis, finding both the user set (represented as vertices) and their vertex correspondence is crucial for tasks like user profiles matching \cite{fan2012}. 
Existing graph alignment analysis and algorithms are not well-suited for these scenarios, as they are not designed to identify the specific subset of vertices where the small graph pattern is embedded within the larger graph. Therefore, the task of locating a small graph pattern within a larger graph is not a trivial extension of the graph alignment problem.
These challenges motivate the formulation of the \emph{subgraph alignment problem}, which aims to recover the vertex set and/or the vertex correspondence of a small graph pattern embedded within a larger graph.

A closely related problem from theoretical computer science is the subgraph isomorphism problem, which studies whether a given pattern exists as an induced subgraph in a larger base graph. The subgraph isomorphism problem is known to be NP-complete \cite{har82}. Nevertheless, it has been extensively investigated over the past decades in both theory and practice due to its wide range of applications. Several influential works, including efficient heuristics and algorithms such as QuickSI \cite{sha08}, GraphQL \cite{he10}, have been developed to address subgraph alignment tasks in real-life applications. 

Another closely related problem is the graph isomorphism problem, which belongs to NP, but is not known to be either in P or NP-complete. The current state of the art of the graph isomorphism problem for all graphs is Babai's quasi-polynomial time algorithm in \cite{bab16}. Moreover, several studies have observed that the average-case complexity \cite{mck81,sch76,sto19} of the graph isomorphism problem under random graph models can be polynomial time, suggesting that an average-case complexity analysis can yield more optimistic results than a worst-case complexity analysis. This motivates us to study the subgraph isomorphism problem from a similar average-case perspective. While it is well-understood that the subgraph isomorphism problem is computationally intractable in the worst-case, its average-case behavior remains far less explored. Therefore, our work (subgraph alignment problem) focuses on analyzing the conditions under which a random subgraph can be reliably located within a base graph, by allowing a vanishing probability of error as the graph size grows.

To investigate the subgraph alignment problem, we introduce the \emph{Erd\"os-R\'enyi subgraph pair model}. For a pair of graphs $(G, H_\pi)$, generated from the Erd\"os-R\'enyi subgraph pair model, we first generate $G$ to be the base Erd\"os-R\'enyi graph on vertex set $[n]$ according to the Erd\"os-R\'enyi measure. Then we choose $m$ vertices from $[n]$ uniformly at random, let $\mathcal{S}$ denote the collection of the chosen vertices, and let $H$ be the induced subgraph on the vertices $\mathcal {S}$. Next, the graph $H$ is anonymized by applying a random bijection $\pi$ from the vertices $\mathcal S$ to $[m]$, and the anonymized graph is denoted as $H_\pi$. Under this formulation, one typical goal is to achieve \emph{exact set recovery}, i.e., recovering the set $\mathcal S$. Moreover, conditioning on having exactly recovered $\mathcal S$, another typical goal is to achieve \emph{exact permutation recovery}, i.e., recovering the set $\mathcal S$ and the permutation $\pi$ at the same time.

Under our Erd\"os-R\'enyi subgraph pair model, we use a brute-force estimator for estimating the vertex set $\mathcal S$ and the bijection $\pi$, and it turns out that the brute-force estimator is equivalent to the maximum a posterior (MAP) estimator. We establish the achievability and converse results for both exact set recovery and exact permutation recovery. Moreover, under certain mild conditions, we establish the following asymptotically tight thresholds for achievability and converse results for exact set recovery \mic{and exact permutation recovery respectively} as $n \to \infty$ (See Theorem \ref{thm: ac pair for set recovery} \mic{and Theorem \ref{thm: ac pair for exact permutation recovery}} for the formal statement).

The main contributions of this paper are summarized as follows.

\begin{enumerate}
    \item \textbf{\emph{Model Formulation.}} We propose the Erd\"os-R\'enyi subgraph pair model, which integrates ideas from the subgraph isomorphism problem and allows randomness in the subgraph, rather than assuming a fixed planted structure. This model formulation enables us to locate a random subgraph structure within a larger random base graph. Moreover, we may also be interested in recovering the vertex labels of the subgraph once it is located, which is often considered in certain applications.

    \item \textbf{\emph{Information theoretic limits.}} We establish achievability and converse results on both recovery criteria: (i) locating the random subgraph; (ii) simultaneously locating the random subgraph and aligning its vertex labels. Our results are tight in certain parameter regimes. 
    
    \item \textbf{\emph{Proof techniques.}} The main techniques for the achievability results are based on probabilistic methods that are widely used in the random graph theory. In particular, we adopt the first-moment method and typicality analysis. For the converse results, we proposed a novel information-theoretic approach to study the infeasibility of vertex-subset recovering, which may be of independent interest.
\end{enumerate}

\subsection{Related Work}

The literature on graph alignment and its extensions is vast. Therefore, we only highlight works most relevant to subgraph alignment, which aligns with the theme of the present paper.

One classical problem in random graph theory that is related to subgraph alignment is the subgraph containment problem, which seeks to determine the threshold at which a fixed pattern $H$ is contained in an Erd\"os-R\'enyi graph. In the 1960s, Erd\"os and R\'enyi established the threshold for balanced graphs in \cite{erd60}. Decades later, Bollob\'as extended these results to a threshold for any graph $H$ in \cite{bol11}. The key difference between subgraph containment and subgraph alignment is that the latter requires the subgraph to have the exact structural correspondence to the pattern graph, while subgraph containment only requires inclusion. Consequently, subgraph alignment imposes a stricter condition. In many applications, such as protein-protein interaction networks in bioinformatics \cite{she12}, it is often necessary to find subgraphs that exactly have the given structure. This motivates the study of the subgraph alignment problem as a natural and more restrictive generalization of the subgraph containment problem.

Another interesting problem in random graph theory is the planted clique problem, which is also closely connected to subgraph alignment. In this problem, one starts with an Erd\"os-R\'enyi random graph with edge probability $p = \frac{1}{2}$ and plants a clique of size $k$ on an arbitrary subset of vertices in the base graph. The objective is to identify the location of the planted clique given only the resulting graph. Both the information-theoretic and algorithmic aspects of the planted clique problem have been studied extensively. It is known that the planted clique problem is infeasible if $k \leq 2\log n$ \cite{bol76}. From an algorithmic perspective, the AKS (Alon-Krevich-Sudakov) spectral algorithm in \cite{alo98} succeeds in detecting and recovering the planted clique for $k \geq n^{\frac{1}{2}+\epsilon}$ in polynomial time for any $\epsilon = \Omega(1)$. On the other hand, it is conjectured that the planted clique problem is computationally hard (specifically, NP-Hard) if $k \leq n^{\frac{1}{2}- \epsilon}$.

The main limitation of the planted clique problem is that it considers only the clique case, i.e., the target subgraph is always a complete graph\mic{\cite{aro09, kuv95, alo98}}. Therefore, some studies generalize the planted clique problem to the planted subgraph problem, where the objective is to identify the location of a planted subgraph in a larger base graph, given only the resulting graph and the underlying planted pattern. A tight information-theoretic threshold for the planted subgraph problem has been established in \cite{mos23, lee25} under a mild density assumption on the subgraphs. However, these studies still assume the subgraph patterns are chosen then planted on the base graph. To the best of our knowledge, there are no existing results on the information-theoretic limits on subgraph recovery where the subgraph itself contains inherent randomness.

\subsection{Overview of the Paper}

This paper is structured as follows. We conclude this introductory section by presenting some of the notation used throughout the paper. In Section \ref{sec: problem formulation}, we formulate the subgraph alignment problem by introducing the proposed Erd\"os-R\'enyi subgraph pair model and the corresponding recovery criterion. In Section \ref{sec: main results}, we state the main results concerning the information-theoretic limits of the subgraph alignment problem, and provide the proofs of the main theorems in Section \ref{sec: proof}. In Section \ref{sec: conclusion}, we conclude the paper and discuss potential future research directions. Finally, the appendix contains some additional proofs.

\subsection{Notation}

In this paper, we use $G \sim \mathrm{ER}(n,p)$ to denote the Erd\"os-R\'enyi model on vertex set $[n] \triangleq \{1,2,\ldots, n\}$, where for any two vertices $1 \leq u,v \leq n$, the edge $\{u,v\}$ is generated i.i.d.~$\mathrm{Bern}(p)$. 

Denote the complement graph of $G$ by $G^c$.

Given a graph $G = (V_G, E_G)$, define the following quantities associated with $G$:\footnote{When $G$ is a random graph, these associated quantities are random variables.}
\begin{itemize}
    \item $v_G = |V_G|$ as the number of vertices in $G$;
    \item $e_G = |E_G|$ as the number of edges in $G$;
    \item $\mathrm{Aut}_G = \left\{\text{bijective }\pi : V \to V \mid  \{u,v\} \in E_G \Leftrightarrow\{\pi(u), \pi(v)\} \in E_G \right\}$ as the number of automorphism of $G$.
\end{itemize}

For a vertex subset $\mathcal S \subseteq [n]$, let $G[\mathcal S]$ denote the induced subgraph on $\mathcal S$. That is, the vertex set of $G[\mathcal S]$ is $\mathcal S$, and for $u, v \in \mathcal S$, $u$ and $v$ are adjacent in $G[\mathcal S]$ if and only if they are adjacent in $G$.

For two graphs $G_1 = (V_1, E_1)$ and $G_2 = (V_2, E_2)$, we write $G_1 \cong G_2$ if $G_1$ is isomorphic to $G_2$, i.e., there exists a bijection $\pi : V_1 \to V_2$ such that $\{u,v\} \in E_1$ if and only if $\{\pi(u), \pi(v)\} \in E_2$. 

We use the standard order notation:
\begin{table}[h!]
    \centering
    \renewcommand{\arraystretch}{1.5}
    \begin{tabular}{c@{\hskip 1cm}c} 
        \toprule
        \textbf{Notation} & \textbf{Definition} \\ 
        \hline 
        $f(n) = \omega(g(n))$ & $\lim_{n \to \infty} \frac{f(n)}{g(n)} = \infty$ \\ 
        $f(n) = o(g(n))$ & $\lim_{n \to \infty} \frac{f(n)}{g(n)} = 0$ \\ 
        $f(n) = O(g(n))$ & $\limsup_{n \to \infty} \frac{f(n)}{g(n)} < \infty$ \\ 
        $f(n) = \Omega(g(n))$ & $\liminf_{n \to \infty} \frac{f(n)}{g(n)} > 0$ \\ 
        $f(n) = \Theta(g(n))$ & $f(n) = O(g(n)) \text{ and } f(n) = \Omega(g(n))$ \\ 
        \bottomrule
    \end{tabular}
\end{table}

Unless stated otherwise, all standard order notation is used under the implicit assumption that $n \to \infty$. Sometimes, we also write $f(n) \asymp g(n)$ if $f(n) = \Theta(g(n))$. All logarithms are assumed to be base $e$, $\log(\cdot) = \log_e(\cdot)$. We use $h(p) = h_e(p) = -p \log p - (1-p) \log (1-p)$ to denote the binary entropy function base $e$. We use $\mathbf{H}(\cdot)$ to denote both the structural entropy of a random structure and the Shannon entropy base $e$. We use $D_{\mathrm{KL}}(\cdot \| \cdot)$ to denote the Kullback-Leibler divergence (relative entropy) base $e$.

\section{Problem Formulation} \label{sec: problem formulation}

In this section, we provide a mathematical formulation for the subgraph alignment problem. We describe the Erd\"os-R\'enyi subgraph pair model in Section \ref{sec: subgraph model}. Under this model formulation, we formally define the subgraph alignment problem in Section \ref{sec: recovery criterion}.

\subsection{Erd\"os-R\'enyi Subgraph Pair Model} \label{sec: subgraph model}

For positive integers $m,n$ such that $m < n$, and real number $p\in [0,1]$, we define the \emph{Erd\"os-R\'enyi subgraph pair model} $(G, H_\pi) \sim \mathcal G(n,m,p)$ by the following steps:

\begin{enumerate}[label = Step \arabic*., leftmargin=*]
    \item (Base graph generating) Let $G \sim \mathrm{ER}(n,p)$ be the Erd\"os-R\'enyi base graph on the vertex set $[n]$.
    \item (Subgraph extracting) Choose $m$ vertices from $[n]$ uniformly at random, let $\mathcal S$ denote the collection of chosen vertices. Let $H = G[\mathcal S]$ be the induced subgraph of $G$ on $\mathcal S$.
    \item (Label anonymizing) Relabel the vertices of $H$ using a random bijection $\pi: \mathcal S \to [m]$, where $\pi$ is chosen uniformly over all $m!$ bijections from $\mathcal S$ to $[m]$. Denote the relabeled graph as $H_\pi$.
\end{enumerate}

Without loss of generality, we assume that the stochasticity in choosing $\mathcal S$ from $[n]$ is independent of the stochasticity in choosing the random bijection $\pi$ from all possible bijections from $\mathcal S$ to $[m]$. An illustration of the model is given in Figure \ref{fig: model illustration}. 

\begin{figure}[htbp]
\centering
\begin{tikzpicture}[scale=0.75, every node/.style={scale=0.9}]

\node at (0.7, 6.0) {$G \sim \mathrm{ER}(n,p)$};

\draw (-2.5, 4) circle (0.5); \node at (-2.5,4) {1}; 
\draw (-0.5, 5) circle (0.5); \node at (-0.5,5) {2}; 
\draw (2, 4.5) circle (0.5); \node at (2,4.5) {3};  
\draw (4, 4) circle (0.5); \node at (4,4) {4};  
\draw (-2, 2) circle (0.5); \node at (-2,2) {5};  
\draw (1, 2) circle (0.5); \node at (1,2) {6};  
\draw (0, 0) circle (0.5); \node at (0,0) {7};  
\draw (2, 0) circle (0.5); \node at (2,0) {8};  
\draw (3.8, 0.5) circle (0.5); \node at (3.8,0.5) {9};  

\draw[-] (-2.4, 3.52) -- (-2.18, 2.46); 
\draw[-] (-0.0, 4.9) -- (1.5, 4.6); 
\draw[-] (2.5, 4.45) -- (3.52, 4.1); 
\draw[-] (1.85, 4.03) -- (1.2, 2.47); 
\draw[-] (2, 4) -- (2, 0.5); 
\draw[-] (3.84, 3.54) -- (2.21, 0.44); 
\draw[-] (-1.65, 1.65) -- (-0.35, 0.35); 
\draw[-] (0.8, 1.55) -- (0.19, 0.45); 
\draw[-] (1.2, 1.55) -- (1.79, 0.44); 

\draw[->, thick] (5.2, 2.5) -- (8.2, 2.5);
\node at (7-0.3, 3.0) {extract $\mathcal{S} \subseteq [n]$};

\node at (11-0.3, 4.5) {$H = G[\mathcal{S}]$};

\draw (10-0.3-0.2, 3.5) circle (0.5); \node at (10-0.3-0.2,3.5) {\textcolor{red}{3}};
\draw (12-0.3-0.2, 3.5) circle (0.5); \node at (12-0.3-0.2,3.5) {\textcolor{blue}{4}};
\draw (10-0.3-0.2, 1.5) circle (0.5); \node at (10-0.3-0.2,1.5) {\textcolor{olive}{6}};
\draw (12-0.3-0.2, 1.5) circle (0.5); \node at (12-0.3-0.2,1.5) {\textcolor{magenta}{8}};

\draw[-] (10.5-0.3-0.2, 3.5) -- (11.5-0.3-0.2, 3.5); 
\draw[-] (10-0.3-0.2, 3) -- (10-0.3-0.2, 2); 
\draw[-] (10.4-0.3-0.2, 3.2) -- (11.6-0.3-0.2, 1.8); 
\draw[-] (12-0.3-0.2, 3) -- (12-0.3-0.2, 2); 
\draw[-] (10.5-0.3-0.2, 1.5) -- (11.5-0.3-0.2, 1.5); 

\draw[->, thick] (13.5-0.3-0.5, 2.5) -- (16-0.3, 2.5);
\node at (14.7-0.3-0.2, 3.5) {vertex bijection $\pi$};
\node at (14.7-0.3-0.2, 3.0) {from $\mathcal{S}$ to $[m]$};

\node at (18-0.3, 4.5) {$H_\pi$};

\draw (17-0.3, 3.5) circle (0.5); \node at (17-0.3,3.5) {\textcolor{red}{1}};
\draw (19-0.3, 3.5) circle (0.5); \node at (19-0.3,3.5) {\textcolor{blue}{2}};
\draw (17-0.3, 1.5) circle (0.5); \node at (17-0.3,1.5) {\textcolor{olive}{3}};
\draw (19-0.3, 1.5) circle (0.5); \node at (19-0.3,1.5) {\textcolor{magenta}{4}};

\draw[-] (17.5-0.3, 3.5) -- (18.5-0.3, 3.5); 
\draw[-] (17-0.3, 3) -- (17-0.3, 2); 
\draw[-] (17.4-0.3, 3.2) -- (18.6-0.3, 1.8); 
\draw[-] (19-0.3, 3) -- (19-0.3, 2); 
\draw[-] (17.5-0.3, 1.5) -- (18.5-0.3, 1.5); 

\end{tikzpicture}   
\caption{Example of Erd\"os-R\'enyi subgraph pair model with $(n,m,p) = (9,4,0.3)$: $G \sim \mathrm{ER}(n,p)$ is generated according to the Erd\"os-R\'enyi measure. Induced subgraph $H$ is obtained through the vertex subset $\mathcal S = \{3,4,6,8\}$. Anonymized subgraph $H_\pi$ is obtained through applying the bijection $\pi: \mathcal S \to [m]$ defined as $\pi(3) = 1$, $\pi(4) = 6$, $\pi(6) = 3$, $\pi(8) = 4$. Note that the receiver will only be given the base graph $G$ and the pattern $H_\pi$, without knowing the bijection $\pi$.}  
\label{fig: model illustration}
\end{figure}
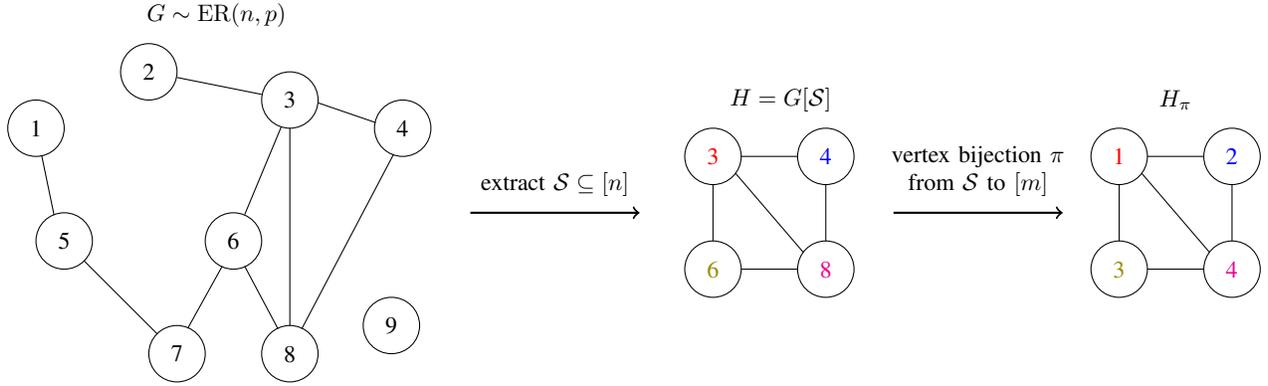

Throughout the remainder of the paper, we work on sequences of (un-subsampled) Erd\"os-R\'enyi subgraph pair model\footnote{A more generalized Erd\"os-R\'enyi subgraph pair model is defined in \cite{shiu25}} with parameters $(n,m(n),p(n))$ such that $\mathbb{N} \ni m(n) < n$ and $p(n) \in [0,1]$ for all $n \in \mathbb{N}$. We are particularly interested in the asymptotic behavior as $n \to \infty$.

\subsection{Recovery Criterion} \label{sec: recovery criterion}

The goal of the subgraph alignment problem is to design estimators $\hat{\mathcal S}(G, H_\pi)$ for recovering the set of vertices $\mathcal S$ and $\hat{\pi}(G, H_\pi)$ for recovering the bijection $\pi$ as functions of $G$ and $H_\pi$. We now formally define the criteria in the subgraph alignment problem.

\textbf{Criterion 1.} (Exact Set Recovery) We say \emph{exact set recovery} is achievable \emph{with high probability (w.h.p.)} if there exists $\hat{\mathcal S}(G, H_\pi)$ such that 
\begin{align*}
    \lim_{n \to \infty} \Pr (\mathcal S = \hat{\mathcal S}) = 1.
\end{align*}

\textbf{Criterion 2.} (Exact Permutation Recovery) We say that \emph{exact permutation recovery} is achievable w.h.p. if there exists $\hat{\mathcal S}(G, H_\pi)$ and $\hat{\pi}(G, H_\pi)$ such that 
\begin{align*}
    \lim_{n\to\infty} \Pr(\mathcal S = \hat{\mathcal S}, \pi = \hat{\pi}) = 1.
\end{align*}

Besides the above criteria, a more general list of recovery criteria for the subgraph alignment problem is described in \cite{shiu25}. In this paper, we will focus on the achievability and converse on both exact set recovery (Criterion 1) and exact permutation recovery (Criterion 2) for the Erd\"os-R\'enyi subgraph pairs $(G, H_\pi) \sim \mathcal G(n,m,p)$. Other cases are left as directions for future research.

\section{Main Results} \label{sec: main results}
In this section, we highlight the main results. Theorem \ref{thm: set recovery achievability} and Theorem \ref{thm: perm recovery achievability} provide an achievable region for exact set recovery and exact permutation recovery, respectively. In Theorem \ref{thm: converse}, \mic{we establish the converse region for exact set recovery. We summarize the achievability-converse pairs and establish the tight threshold under certain conditions for exact set recovery and exact permutation recovery in Theorem \ref{thm: ac pair for set recovery} and Theorem \ref{thm: ac pair for exact permutation recovery} respectively}. The proofs of the achievability and converse results are given in Section \ref{sec: proof}.

\mic{In this work, we are interested in the information-theoretic limits of the exact subgraph alignment problem. Accordingly, we assume access to unlimited computational power and focus on whether reliable recovery is possible in principle. In particular, we use the theoretically most powerful estimator -- the brute-force estimator, which eventually boils down to exhaustively checking all possible subsets and labelings and outputs those that are consistent with the observations. In the following, for an Erd\"os-R\'enyi random graph $G \sim \mathrm{ER}(n,p)$ and $m \in \mathbb{N}$ with $m < n$, define
\begin{align*} 
    \Theta := \{(S, \sigma) : S \subseteq V \text{ with } |S| = m, \sigma: S \to [m] \text{ bijective}\}
\end{align*}
to be the set of all possible configurations in the Erd\"os-R\'enyi subgraph pair model described in Section \ref{sec: subgraph model}.}

\mic{
The brute-force estimator can be formulated as the following maximization problem.
}

\textcolor{black}{
    \begin{definition}[Brute-force Estimator] \label{defn: brute-force}
        Consider the Erd\"os-R\'enyi subgraph pair $(G, H_\pi) \sim \mathcal G(n,m,p)$. The output of the brute-force estimator is given by the set
        \begin{align} \label{brute force output}
            \mathcal O_{\mathrm{BF}} := \argmax_{(S, \sigma) \in \Theta} \mathbf{1}\{\sigma(G[S]) = H_\pi\}.
        \end{align}
    \end{definition}
}

\mic{Another commonly used estimator for graph recovery problems is the maximum a posteriori (MAP) estimator, which can be formulated as the following maximization problem.}

\textcolor{black}{
    \begin{definition}[MAP Estimator] \label{defn: map}
        Consider the Erd\"os-R\'enyi subgraph pair $(G, H_\pi) \sim \mathcal G(n,m,p)$. Let $(\mathcal S, \Sigma) \sim \mathrm{Unif}(\Theta)$, and let $(\mathcal G, \mathcal H)$ be the random variables that describe the graph realization from the Erd\"os-R\'enyi subgraph pair model described in Section \ref{sec: subgraph model}. The output of the MAP estimator is given by the set
        \begin{align} \label{MAP output}
            \mathcal O_{\mathrm{MAP}} := \argmax_{(S, \sigma) \in \Theta} \Pr\left(\mathcal S = S, \Sigma = \sigma \mid \mathcal G = G, \mathcal H = H_\pi\right).
        \end{align}
    \end{definition}
}

\mic{Under the Erd\"os-R\'enyi subgraph pair model, this brute-force estimator turns out to coincide with the MAP estimator, which is formally stated in the following proposition.}

\textcolor{black}{
    \begin{proposition} \label{prop: equivalence}
        Consider the Erd\"os-R\'enyi subgraph pair $(G, H_\pi) \sim \mathcal G(n,m,p)$. Let $\mathcal O_{\mathrm{BF}}$ and $\mathcal O_{\mathrm{MAP}}$ be the outputs from the brute-force estimator and MAP estimator in (\ref{defn: brute-force}) and (\ref{defn: map}), respectively. Then $\mathcal O_{\mathrm{BF}} = \mathcal O_{\mathrm{MAP}}$. 
    \end{proposition}
}
\begin{proof}
    See Appendix \ref{app: equivalence}.
\end{proof}

\begin{algorithm}
    \caption{A Brute-force Algorithm for Estimating $\hat{\mathcal S}$ and $\hat{\pi}$} 
    \label{alg: brute-force algo} 
    \SetKw{Compute}{compute}
    \SetKwFor{Foreach}{for each}{do}{end for}
    \SetKwInOut{Input}{input}
    \SetKwInOut{Output}{output}
    \SetKw{Break}{break}
    \SetKw{Add}{add}
    
    \Input{Graphs $G$, $H_\pi$.}

    \Output{\mic{$\mathcal O$, a collection of set and permutation estimate $(\hat{\mathcal S},\hat{\pi})$.}}
    
    Let $A_{H_\pi}$ be the adjacency matrix of $H_\pi$

    $\mathcal O \gets \emptyset$

    \Foreach{size-$m$ subset $\mathcal S \subseteq [n]$}{
        Let $A_\mathcal S$ be the adjacency matrix of $G[\mathcal S]$
        
        \Foreach{bijection $\sigma$ from $\mathcal S$ to $[m]$}{
            Define $A_{\mathcal S_\sigma}$ as follows
            
            \For{$i, j \in \mathcal S$}{
                $A_{\mathcal S_\sigma}(\sigma(i), \sigma(j)) \gets A_{\mathcal S}(i,j)$
            }
        }
        
        \If{$A_{H_\pi} = A_{\mathcal S_\sigma}$} {
            \textcolor{black}{\Add $(\mathcal S, \sigma)$ to $\mathcal O$}
        }
    }
\end{algorithm}

\mic{The Erd\"os-R\'enyi subgraph pair model ensure that the output set $\mathcal O$ from Algorithm \ref{alg: brute-force algo} is always non-empty. If $\mathcal O$ contains at least two candidates, we design a rule for tie-breaking.}

\begin{theorem}[Exact Set Recovery Achievability] \label{thm: set recovery achievability} 
    Consider the Erd\"os-R\'enyi subgraph pair $(G, H_\pi) \sim \mathcal G(n,m,p)$. If 
    \begin{align} \label{set achievability condition}
        \frac{m}{2}h(p) - \log n \to \infty,
    \end{align}
    then the MAP estimator $\hat{\mathcal S}$ from Algorithm \ref{alg: brute-force algo} achieves exact set recovery w.h.p.
\end{theorem}

\begin{theorem}[Exact Permutation Recovery Achievability] \label{thm: perm recovery achievability}
    Consider the Erd\"os-R\'enyi subgraph pair $(G, H_\pi) \sim \mathcal G(n,m,p)$. If 
    \begin{numcases}{}
        \frac{m}{2}h(p) - \log n \to \infty \label{permutation achievability condition 1} \\
        mp - \log m \to \infty, \label{permutation achievability condition 2}
    \end{numcases}
    then the MAP estimator $(\hat{\mathcal S}, \hat{\pi})$ from Algorithm \ref{alg: brute-force algo} achieves exact permutation recovery w.h.p.
\end{theorem}

\mic{A natural question is whether both conditions (\ref{permutation achievability condition 1}) and (\ref{permutation achievability condition 2}) in Theorem \ref{thm: perm recovery achievability} are necessary. The corresponding discussion is postponed to Section \ref{sec: perm recovery achievability proof}.}

\begin{theorem}[Converse] \label{thm: converse}
    Consider the Erd\"os-R\'enyi subgraph pair $(G, H_\pi) \sim \mathcal G(n,m,p)$. If
    \begin{align}
        \quad & \frac{m}{2}h(p) - \log \frac{n}{m} \to -\infty, \label{new converse region 1}
    \end{align}
    then no algorithm achieves exact set recovery w.h.p.
\end{theorem}

\begin{remark}
    A converse for exact set recovery is immediately a converse for exact permutation recovery.
\end{remark}

\textcolor{black}{
We summarize the achievability-converse results for exact set recovery and exact permutation recovery, respectively, in the following theorems, and we justify the correctness in Appendix \ref{app: justification for ac pairs}.
\begin{theorem}[Achievability-Converse pair for Exact Set Recovery] \label{thm: ac pair for set recovery}
    Consider the Erd\"os-R\'enyi subgraph pair $(G, H_\pi) \sim \mathcal G(n,m,p)$. \\ 
    \noindent\underline{Achievability:} If 
    \begin{align} \label{tight threshold achievability}
        \frac{m}{2}h(p) - \log n \to \infty,
    \end{align}
    then the MAP algorithm achieves exact set recovery w.h.p. \\
    \noindent\underline{Converse:} If 
    \begin{align}
        \frac{m}{2}h(p) - \log \frac{n}{m} \to -\infty,
    \end{align}
    then no algorithm achieves exact set recovery w.h.p. \\
    \noindent Furthermore, if any of the following conditions if any of the following conditions are satisfied:
    \begin{align}
        \mathrm{(i)} \quad & \log m = o(\log n), \qquad \label{tight threshold case 1} \text{ or } \\
        \mathrm{(ii)} \quad & mp - \log m\to \infty, \label{tight threshold case 2}
    \end{align}
    then the converse condition can be improved to 
    \begin{align}
        \frac{m}{2}h(p) - \log n \to -\infty, \label{tight threshold converse}
    \end{align}
\end{theorem}
\begin{theorem}[Achievability-Converse pair for Exact Permutation Recovery] \label{thm: ac pair for exact permutation recovery}
    Consider the Erd\"os-R\'enyi subgraph pair $(G, H_\pi) \sim \mathcal G(n,m,p)$. \\ 
    \noindent\underline{Achievability:} If 
    \begin{align}  \label{ac pair achievability permutation recovery}
        \frac{m}{2}h(p) - \log n \to \infty \qquad \text{ and } \qquad mp - \log m \to \infty,
    \end{align}
    then the MAP algorithm achieves exact permutation recovery w.h.p. \\
    \noindent\underline{Converse:} If 
    \begin{align} \label{ac pair converse permutation recovery}
        \frac{m}{2}h(p) - \log \frac{n}{m} \to -\infty \qquad \text{ or }  \qquad mp - \log m \to -\infty,
    \end{align}
    then no algorithm achieves exact permutation recovery w.h.p. \\
    \noindent Furthermore, if $\log m = o(\log n)$, then the converse condition can be improved to 
    \begin{align*}
        \frac{m}{2}h(p) - \log n \to -\infty.
    \end{align*}
\end{theorem} 
}

\begin{figure}
    \centering
    \includegraphics[width=0.35\linewidth]{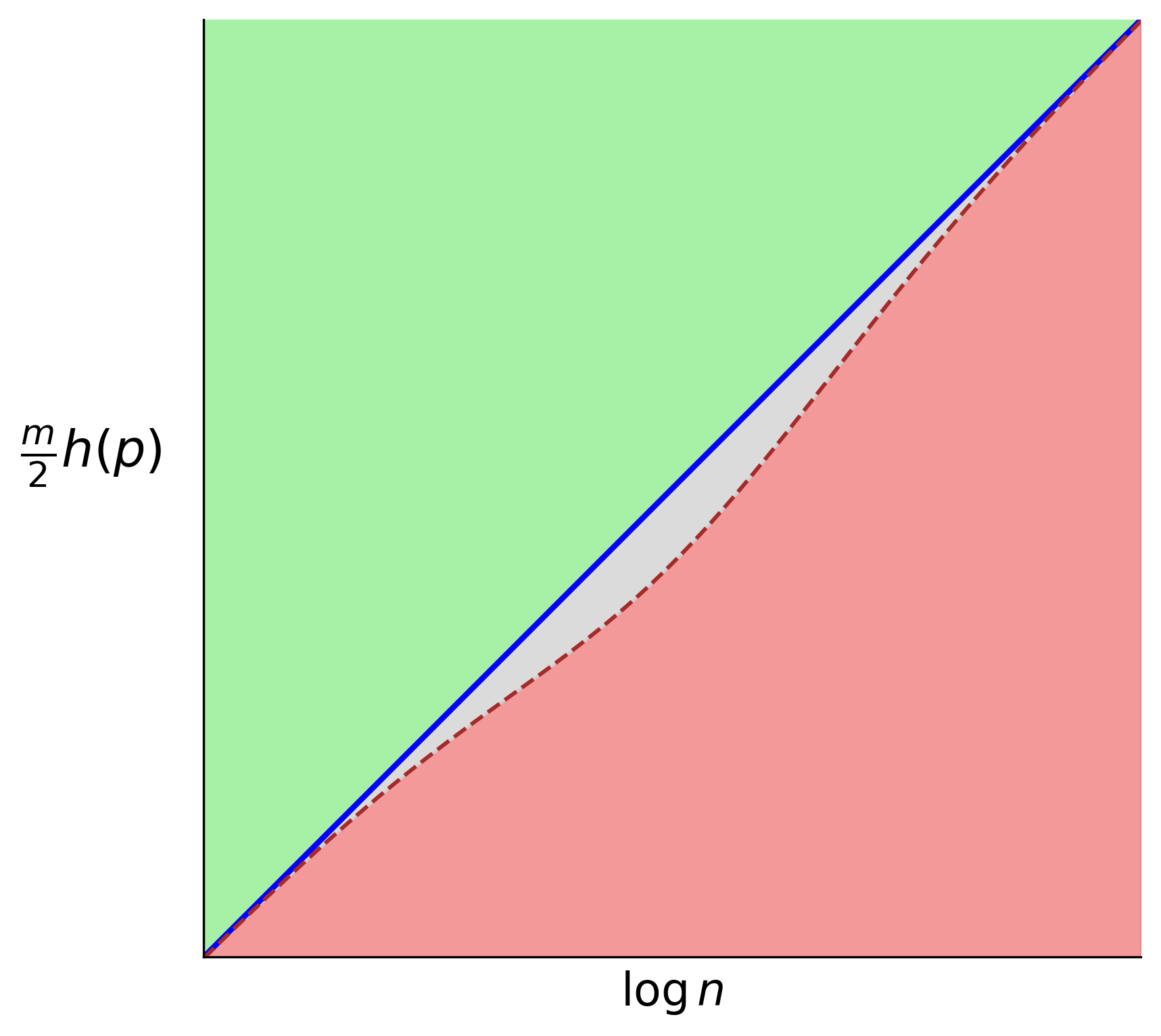}
    \caption{The green region represents the information‑theoretically achievable regime, while the red region corresponds to the non‑achievable regime. The grey region indicates cases whose achievability remains unknown. It is worth noting that the achievability and converse conditions coincide to form a tight threshold under certain conditions.}
    \label{fig: region}
\end{figure}

Before ending this section, we state two useful lemmas, which allow us to focus on a smaller class of graph sequences.

\begin{lemma}\label{lemma: complement graph}
    Let $G \sim \mathrm{ER}(n,p)$, and let $H, H_\pi$ be the subgraphs of $G$ obtained from Section \ref{sec: subgraph model}. Let $\tilde{H}$ be the subgraph of $G^c$ induced by the same $\mathcal S$, then $H^c_\pi \cong \tilde{H}$, where $H^c_\pi$ is the complement graph of $H_\pi$ on $[m]$.
\end{lemma}

\begin{proof}
    See Appendix \ref{app: proof of complement graph lemma}.
\end{proof}

\begin{lemma}\label{lemma: p <= 1/2}
    Let $G \sim \mathrm{ER}(n,p)$. Then its complement graph $G^c \sim \mathrm{ER}(n,1-p)$.
\end{lemma}

\begin{proof}
    See Appendix \ref{app: proof of p <= 1/2}.
\end{proof}

Hence, throughout the remainder of the paper, we further assume that $p(n) \in [0, 1/2]$ for all $n \in \mathbb{N}$ by combining the arguments in Lemma \ref{lemma: complement graph} and Lemma \ref{lemma: p <= 1/2}.

\section{Proof of Main Theorems} \label{sec: proof}

In this section, we prove our main results from Section \ref{sec: main results}. Specifically, we prove the two achievability theorems, including Theorem \ref{thm: set recovery achievability} and Theorem \ref{thm: perm recovery achievability} in Section \ref{sec: set recovery achievability proof} and Section \ref{sec: perm recovery achievability proof} respectively. We then establish the converse, Theorem \ref{thm: converse} in Section \ref{sec: converse proof}. \mic{In Appendix \ref{app: comparison}, we compare the results presented in Section \ref{sec: main results} with those presented in \cite{shiu25}.} Additionally, we justify the cases when our achievability and converse results form a tight threshold for exact set recovery in Section \ref{sec: threshold proof}.

\subsection{Proof of Achievability Theorem for Set Recovery} \label{sec: set recovery achievability proof}

In this section, we prove Theorem \ref{thm: set recovery achievability}. The approach for proving achievability for set recovery relies on the \mic{first-moment} method. A key observation for set recovery is that under Algorithm \ref{alg: brute-force algo}, exact set recovery is possible if and only if $H$ is the unique induced copy in $G$ with high probability. Therefore, we first define a random variable $X_H$ that counts the number of induced subgraphs in $G$ that are isomorphic to $H$ for a given $H$, then calculate the expectation of $X_H$ in Lemma \ref{lemma: close form expectation}. 

\begin{definition}\label{defn:XH}
    Let $H$ be an induced subgraph of $G$, and let $N = \binom{n}{v_H}$. Define the random variable $X_H$ as the number of subsets of $[n]$ that induce a subgraph isomorphic to $H$. In particular, if $\{\mathcal S_1, \ldots, \mathcal S_N\}$ is an enumeration of all subsets of vertices of $[n]$ with size $v_H$, then 
    \begin{align}
        X_H \triangleq \sum_{j=1}^{N} I_j,\label{X_H sum} 
    \end{align}
    where $I_j$ is the binary random variable that indicates whether the induced subgraph of $G$ on $\mathcal S_j$ is isomorphic to $H$, i.e., $G[\mathcal S_j] \cong H$.
\end{definition}

\begin{lemma}\label{lemma: close form expectation}
    Let $H$ be an induced subgraph of $G$, and let $X_H$ be defined in Definition \ref{defn:XH}, then
    \begin{align}
        \mathbb{E}(X_H \mid H) = \binom{n}{v_H} \frac{v_H!}{\mathrm{Aut}_H}p^{e_H}(1-p)^{\binom{v_H}{2} - e_H}. \label{expectation}
    \end{align}
\end{lemma}

The proof of Lemma \ref{lemma: close form expectation} is provided in Appendix \ref{app: proof of close form}.

\begin{definition} \label{defn: typical graph}
    Let $n \in \mathbb{N}$, define the set of typical graphs
    \begin{align*}
        \mathcal T_{\varepsilon}^n = \left\{\text{graph } G \text{ on } n \text{ vertices }: \left| e_G - \binom{n}{2}p \right| \leq \varepsilon \binom{n}{2}p\right\},
    \end{align*}
    where $\varepsilon > 0$. In other words, a typical graph $G \in \mathcal T_{\varepsilon}^n$ satisfies
    \begin{align} \label{edge approx for typical graph}
        (1 - \varepsilon) \binom{n}{2}p \leq e_G \leq (1 + \varepsilon) \binom{n}{2}p.
    \end{align}
\end{definition}

We may now prove Theorem \ref{thm: set recovery achievability}. The proof is done by showing that $\Pr(X_H \geq 2) \to 0$ as $n \to \infty$ using the \mic{first-moment} method. By law of total probability, we can split $\Pr(X_H \geq 2)$ in the following way:

\begin{align}
    \Pr(X_H \geq 2) & = \Pr(H \in \mathcal T^m_{\varepsilon}) \Pr(X_H \geq 2 \mid H \in \mathcal T^m_\varepsilon) + \Pr(H \notin \mathcal T^m_{\varepsilon}) \Pr(X_H \geq 2 \mid H \notin \mathcal T^m_\varepsilon) \notag \\
    & \leq \Pr(X_H \geq 2 \mid H \in \mathcal T^m_\varepsilon) + \Pr(H \notin \mathcal T^m_\varepsilon) \notag \\
    & \leq \sup_{H \in \mathcal T^m_\varepsilon} \Pr(X_H \geq 2 \mid H) + \Pr(H \notin \mathcal T^m_\varepsilon) \label{upper bound average by supremum}\\
    & \leq \frac{1}{2} \underbrace{\sup_{H\in \mathcal T^m_\varepsilon} \mathbb{E}(X_H \mid H)}_{\mathrm{(A1)}} + \underbrace{\Pr(H \notin \mathcal T^m_\varepsilon)}_{\mathrm{(A2)}}, \label{markov upper bound}
\end{align}
where in (\ref{upper bound average by supremum}) we upper bound the average of $\Pr(X_H \geq 2 \mid H = H')$ over $H' \in \mathcal T^m_\varepsilon$ with the supremum , and in (\ref{markov upper bound}) we applied Markov inequality. Now, we would like to show that both (A1) and (A2) vanish as $n \to \infty$ by suitably choosing a $\varepsilon = \varepsilon(n) > 0$.

First, we upper bound the combinatorial coefficients in (\ref{expectation}) by 
\begin{align}
    \binom{n}{v_H} \frac{v_H!}{\mathrm{Aut}_H} = \frac{n!}{m!(n-m)!} \cdot \frac{m!}{\mathrm{Aut}_H} = \frac{n!}{(n-m)!} \leq n^m, \label{simplified combinatorial coefficient}
\end{align}
where we used the fact that $\mathrm{Aut}_H \geq 1$ for any graph $H$. Therefore, for any $H \in \mathcal T^m_{\varepsilon}$, using (\ref{expectation}) and (\ref{edge approx for typical graph}) yield
\begin{align}
    \mathbb{E}(X_H \mid H) & \leq n^m p^{e_H} (1-p)^{\binom{m}{2} - e_H} \notag \\
    & = \exp\left(m \log n + e_H \log p + \left(\binom{m}{2} - e_H \right) \log (1-p)\right) \notag  \\
    & \leq \exp\left( m \log n + (1 - \varepsilon) \binom{m}{2}p \log p + \left(\binom{m}{2} - ( 1 + \varepsilon)\binom{m}{2}p\right)\log(1-p)\right) \notag \\
    & = \exp\left(m \log n - \binom{m}{2} h(p) + \varepsilon \binom{m}{2} \left( p\log \frac{1}{p} + p \log \frac{1}{1-p}\right)\right) \notag  \\
    & \leq \exp\left(m \log n - \binom{m}{2} h(p) + \varepsilon \binom{m}{2} h(p)\right) \label{h(p) inequality}\\
    & = \exp\left(m \log n - (1 - \varepsilon)\binom{m}{2} h(p)\right) \label{typical term},
\end{align}
where in (\ref{h(p) inequality}) we used the inequality
\begin{align*}
    p \log \frac{1}{p} + p \log \frac{1}{1-p} \leq p \log \frac{1}{p} + (1-p) \log \frac{1}{1-p} = h(p)
\end{align*}
which holds for all $p \in [0, 1/2]$. Note that (\ref{typical term}) provided a universal bound on $\mathbb{E}(X_H \mid H)$ over all $H \in \mathcal T^m_\varepsilon$, so we also have
\begin{align}
    \sup_{H \in \mathcal T_\varepsilon^m} \mathbb{E}(X_H \mid H) \leq \exp\left(m \log n - (1 - \varepsilon) \binom{m}{2} h(p)\right). \label{uniform bound on typical term}
\end{align}

On the other hand, note that the edge random variable $e_H$ can be regarded as a sum of $\binom{m}{2}$ i.i.d. $\mathrm{Bern}(p)$, so $e_H \sim \mathrm{Binom}\left(\binom{m}{2}, p\right)$ with $\mathbb{E}(e_H) = \binom{m}{2}p$. Hence, we may apply the multiplicative Chernoff bound on (A2) to obtain
\begin{align}
    \Pr(H \notin \mathcal T^m_\varepsilon) & = \Pr\left( \left| e_H - \binom{m}{2} p \right| > \varepsilon \binom{m}{2} p \right) \leq 2 \exp\left(-\frac{\varepsilon^2\binom{m}{2}p}{3}\right). \label{atypical term}
\end{align}
It remains to choose an appropriate $\varepsilon > 0$ such that (\ref{uniform bound on typical term}) and (\ref{atypical term}) vanish simultaneously as $n \to \infty$. 

Pick $\varepsilon = \frac{1}{\sqrt{m\sqrt{p}}} > 0$. We claim that $m\sqrt{p} \to \infty$ under the condition (\ref{set achievability condition}). We split into two cases when (i) $p \neq o(1)$; (ii) $p = o(1)$.

\begin{itemize}
    \item (Case i) Since $p \neq o(1)$, we can find some universal constant $\eta > 0$ such that $p(n) \geq \eta$ for sufficiently large $n$. Therefore,
    \begin{align*}
        m\sqrt{p} \geq m\sqrt{\eta} \to \infty
    \end{align*}
    as $m \to \infty$ is implicit from the condition (\ref{set achievability condition}).
    
    \item (Case ii) Note that there is some $p^{**} > 0$ such that $\sqrt{p} \geq h(p)$ for all $p \in [0, p^{**}]$. Therefore,
    \begin{align*}
        m\sqrt{p} \geq m h(p) \to \infty
    \end{align*}
    as $mh(p) \to \infty$ is again implicit from the condition (\ref{set achievability condition}).
\end{itemize}

Therefore, $\varepsilon = (m\sqrt{p})^{-1/2} = o(1)$. Combining with condition (\ref{set achievability condition}), (\ref{uniform bound on typical term}) vanishes as
\begin{align}
    \exp\left(m \log n - (1- \varepsilon) \binom{m}{2}p\right) \asymp \exp\left( m \left( \log n - (1 - o(1)) \frac{m}{2}h(p)\right)\right) \to 0. \label{typical term vanish}
\end{align}

At the same time, (\ref{atypical term}) also vanishes as
\begin{align}
    2 \exp\left(-\frac{\varepsilon^2\binom{m}{2}p}{3}\right) & \asymp 2 \exp\left(-\frac{\frac{1}{m\sqrt{p}} \cdot\frac{m^2p}{2}}{3}\right) = 2 \exp\left(-\frac{m\sqrt{p}}{6}\right) \to 0, \label{atypical term vanish}
\end{align}
where we again use the fact that $m\sqrt{p} \to \infty$.

Putting (\ref{typical term vanish}) and (\ref{atypical term vanish}) back to (\ref{markov upper bound}), we have $\Pr(X_H \geq 2) \to 0$. That is, among all induced subgraph of size $v_H$ in $G$, if $G[\mathcal S] \cong G[\mathcal S']$, then $\mathcal S = \mathcal S'$ w.h.p. Algorithm \ref{alg: brute-force algo} outputs the first vertex subset $\hat{\mathcal S}$ that satisfies $G[\hat{\mathcal S}] \cong G[\mathcal S]$, it follows that $\Pr(\mathcal S = \hat{\mathcal S}) \to 1$ as $n \to \infty$.

\subsection{Proof of Achievability Theorem for Permutation Recovery} \label{sec: perm recovery achievability proof}

In this section, we prove Theorem \ref{thm: perm recovery achievability}. We begin with an important lemma from \cite{wri97}.

\begin{lemma} \label{lemma: wright}
    Let $G \sim \mathrm{ER}(n,p)$. If $np - \log n \to \infty$, then $G$ has a trivial automorphism group with probability $1 - o(1)$. Conversely, if $np- \log n \to -\infty,$ then the probability that $G$ has a trivial automorphism group is bounded away from 1. 
\end{lemma}

From Theorem \ref{thm: set recovery achievability}, we can recover the correct $\mathcal S = \hat{\mathcal S}$ w.h.p. under condition (\ref{permutation achievability condition 1}). Therefore, using the fact that $H \sim \mathrm{ER}(m,p)$, under condition (\ref{permutation achievability condition 2}), $H$ has a trivial automorphism group with high probability. That is, conditioned on we have recovered a correct $\mathcal S = \hat{\mathcal S}$, there is a unique bijection $\hat{\pi}$ such that $H_\pi = H_{\hat{\pi}}$. Algorithm \ref{alg: brute-force algo} outputs the first bijection $\hat{\pi}$ that satisfies $H_\pi = H_{\hat{\pi}}$, it follows that $\Pr(\mathcal S = \hat{\mathcal S}, \pi = \hat{\pi}) \to 1$ as $n \to \infty$.

{\color{black}
A natural question for Theorem \ref{thm: perm recovery achievability} is whether one of the conditions can be absorbed by the other one. And the answer is no. To show that the conditions are not redundant, we simply construct a triplet sequence $(n,m(n), p(n))$ that satisfies the condition \ref{permutation achievability condition 1} but not condition \ref{permutation achievability condition 2}.

Consider the triplet sequence $(n,m(n),p(n))$ with $m(n) = n/2$ and $p(n) = \frac{2\log \frac{n}{2}}{n} = \frac{\log m}{m}$. Note that $p(n) \to 0$ as $n \to \infty$, we may approximate $h(p)$ by 
\begin{align} \label{binary entropy approximation}
    h(p) = -p\log p -(1-p) \log (1-p) = p\log \frac{1}{p}(1 + o(1)).
\end{align}
For $p(n) = \frac{\log m}{m}$, we may expand it as 
\begin{align*}
    p \log \frac{1}{p} & = \frac{\log m }{m} \log \frac{m}{\log m} \\
    & = \frac{(\log m)^2}{m} - \frac{(\log m)(\log \log m)}{m} \\
    & = \frac{(\log m)^2}{m}(1 - o(1)).
\end{align*}
Therefore, condition (\ref{permutation achievability condition 1}) holds as 
\begin{align*}
    \frac{m}{2}h(p) - \log n & = \frac{m}{2}\frac{(\log m)^2}{m}(1 + o(1)) - \log n \\
    & = \frac{1+o(1)}{2} \left(\log \frac{n}{2}\right)^2 - \log n \\
    & \to \infty.
\end{align*}
But condition (\ref{permutation achievability condition 2}) fails as 
\begin{align*}
    mp - \log m  = m \cdot \frac{\log m }{m}- \log m = 0 \not \to \infty.
\end{align*}

The above example also demonstrates that there is a region where we can recover the set $\hat{\mathcal S}$ correctly w.h.p., but any algorithm cannot recover the permutation $\hat{\pi}$ conditioned on $\hat{\mathcal S} = \mathcal S$ w.h.p.}

\subsection{Proof of Converse Theorem} \label{sec: converse proof}

In this section, we prove Theorem \ref{thm: converse}. Instead of employing a \mic{second-moment} method that is commonly used in many classical random graph theory problems, we adopt an information-theoretic-based approach to establish the converse region. \mic{In information theory, a classical result in lossless compression states that if the compression rate is below the fundamental limit, i.e., the entropy of the source, then no decoding scheme can reconstruct the source with vanishing error probability. It is natural to ask whether an analogous extension exists for non-conventional data types such as lattices\cite{sha53}, graphs, etc. In particular, recent work \cite{cho12} has studied the lossless compression of unlabeled graphs, leading to the notion of \textit{structural entropy}, which plays a role analogous to Shannon entropy by characterizing the fundamental limits of compressing graphs. This motivates us to model the subgraph alignment problem as an equivalent problem in the language of lossless compression and to prove a converse in a similar sense.} We begin by formulating the subgraph alignment problem as a \mic{lossless compression} scenario.

\begin{itemize}
    \item \textbf{Initialization.} Initially, both the encoder and decoder agree on the graph parameters $(n,m,p)$. Also, they share a base graph, which is a realization of the Erd\"os-R\'enyi random graph \mic{$G \sim \mathrm{ER}(n,p)$}.

    \item \textbf{Encoder.} On the encoder side, let $\binom{[n]}{m} = \{ S \subset V_G: |S| = m\}$ be the source to be compressed, it is clear that $\left|\binom{[n]}{m}\right| = \binom{n}{m}$. The source is $\mathcal{S} \sim \mathrm{Unif}\binom{[n]}{m}$.

    Suppose the encoder wants to compress the source $\mathcal S \in \binom{[n]}{m}$, then it sends the unlabeled graph $H = G[\mathcal S]$ to the decoder.\footnote{Since the encoder only compress the unlabeled graph, it is equivalent to the label anonymizing procedure by applying a random permutation on the vertices.}

    \item \textbf{Decoder.} The decoder receives the unlabeled graph (source) $H = G[\mathcal S]$ losslessly. Then it uses an estimator base on observing $G$ and $H$ to a find a vertex set $\hat{\mathcal S} \in \binom{[n]}{m}$ such that $G[\hat{\mathcal S}] \cong H$, and declare $\hat{\mathcal S}$ as the estimation. \footnote{Since the graph is compressed losslessly, such $\hat{\mathcal S}$ always exists.}
\end{itemize}

\mic{An illustration of the lossless compression setup is given in Figure \ref{fig: lossless compression setup}.}

\begin{figure}[htbp]
\centering
\begin{tikzpicture}[scale=0.75, every node/.style={scale=0.9}, >=stealth]

  \node[draw, minimum width=2.5cm, minimum height=1.2cm] (enc) at (0,0) {\large{Encoder}};
  \node[draw, minimum width=2.5cm, minimum height=1.2cm] (dec) at (6,0) {\large{Decoder}};

  \draw[->] (-3,0) -- (enc.west) node[midway, above] {$\mathcal S$};

  \draw[->] (enc.east) -- (dec.west) node[midway, above] {$H = G[\mathcal S]$};

  \draw[->] (dec.east) -- (9,0) node[midway, above] {$\hat{\mathcal S}$};

\end{tikzpicture}
\caption{\mic{The encoder encodes the source $\mathcal S$ by compressing the unlabaled graph $H = G[\mathcal S]$, and the decoder uses both $G$ and $H$ to reconstruct the estimation $\hat{\mathcal S}$. Unlike the general lossless compression setup, the encoder has no freedom to design the encoding function.}}
\label{fig: lossless compression setup}
\end{figure}
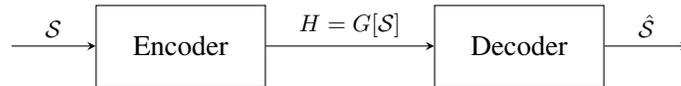

We next introduce the concept of structural entropy, which quantifies the minimum amount of resources required to compress an unlabeled graph losslessly. Structural entropy is analogous to Shannon entropy, but specifically for graph-based sources.

In the following, let $\mathcal G_n$ denote the collection of all labeled graphs on $n$ vertices. We begin by defining graph entropy and structural entropy \cite{cho12}.

\begin{definition}[Graph Entropy] \label{defn: graph entropy}
    Let $P$ be a probability distribution on $\mathcal G_n$. Let $\mathcal G$ be a random graph model with distribution $P$. Then the graph entropy of $\mathcal G$ is defined as 
    \begin{align}\label{graph entropy}
        \mathbf H(\mathcal G)= - \mathbb{E}_{G \sim P}(\log P(G)) = -\sum_{G \in \mathcal G_n} P(G) \log P(G).
    \end{align}
\end{definition}

Next, we introduce the random structure model $\mathcal U$ for the unlabeled version of a random graph model, i.e., the vertices are indistinguishable, and the graphs that are isomorphic in $\mathcal G_n$ belong to the same structure $U \in \mathcal U$. For a given structure $U \in \mathcal U$, the probability of $U$ can be computed as 
\begin{align}
    P(U) = \sum_{G \cong U, G \in \mathcal G_n} P(G). \label{random structure probability}
\end{align}

\begin{definition}[Structural Entropy] \label{defn: structural entropy}
    Let $P$ be a probability distribution on $\mathcal G_n$. Let $\mathcal G$ be a random graph model with distribution $P$. Then the structural entropy of $\mathcal G$ is defined as the entropy of the random structure $\mathcal U$, 
    \begin{align} \label{sturctural entropy}
        \mathbf H(\mathcal U) = - \sum_{U \in \mathcal U} P(U) \log P(U).
    \end{align}
\end{definition}

\begin{theorem}[Structural Entropy of Erd\"os-R\'enyi Model] \label{thm structural entropy of ER}
    Suppose $np - \log n \to \infty$, and let $P$ be the Erd\"os-R\'enyi measure on $\mathcal G_n$. Then the structural entropy of $\mathcal G \sim P$ is 
    \begin{align} \label{structural entropy ER}
        \mathbf H(\mathcal U) = \binom{n}{2}h(p) - \log n! + o(1).
    \end{align}
    For general $n$ and $p$, the structural entropy is upper bounded by 
    \begin{align} \label{general structural entropy upper bound}
        \mathbf H(\mathcal U) \leq \binom{n}{2}h(p).
    \end{align}
\end{theorem}

The structural entropy of a random graph model reflects the number of bits required to capture the minimum information necessary to describe its structure and compress the graph model losslessly. In particular, \cite{cho12} introduces a polynomial time algorithm that demonstrates the number of bits required to compress an unlabeled graph $G$ from an Erd\"os-R\'enyi model $\mathrm{ER}(n,p)$ can be asymptotically quantified by the structural entropy $H(\mathcal U)$ is provided in \cite{cho12}. Specifically, the expected length of the compressed codeword $L(U)$ satisfies 
\begin{align*}
    \mathbb{E}(L(U)) \leq \binom{n}{2} h(p) - n \log n + o(n \log n),
\end{align*}
aligning with the structural entropy.

Now we are ready to establish the converse region. Recall that the goal of the decoder is to estimate the source $\hat{\mathcal S}$ based on the observation $H = G[\mathcal S]$, with the base graph $G$ as the underlying side information. Therefore, it is infeasible for the decoder to recover $\hat{\mathcal S} = \mathcal S$ correctly with high probability if 
\begin{align} \label{infeasibility condition}
    \mathbf{H}(G[\mathcal S] \mid G) < \mathbf{H}(\mathcal S \mid G) = \mathbf{H}(\mathcal S),
\end{align}
where the equality follows from the assumption that the source is chosen independent of the graph generating process. Next, observe that by (\ref{general structural entropy upper bound}) 
\begin{align} \label{information for decoder}
    \mathbf{H}(G[\mathcal S] \mid G) \leq \mathbf{H}(G[\mathcal S]) \leq \binom{m}{2} h(p),
\end{align}
where we used conditioning to reduce entropy for the first inequality.

Also, as we assume the source $\mathcal S \sim \mathrm{Unif}\binom{[n]}{m}$ is uniform over a set of size $\binom{n}{m}$,
\begin{align}
    \mathbf{H}(\mathcal S) = \log \binom{n}{m} \geq m \log \frac{n}{m}.
\end{align}
Now note that the condition (\ref{new converse region 1}) implies that 
\begin{align} \label{intermediate infeasibility condition}
    \frac{m-1}{2} h(p) \leq \frac{m}{2}h(p) \ll \log \frac{n}{m}.
\end{align}
Multiplying both sides by $m$ on (\ref{intermediate infeasibility condition}) leads to the inequality
\begin{align*}
    \mathbf{H}(G[\mathcal S] \mid G) \leq \binom{m}{2} h(p) \ll m \log \frac{n}{m} \leq \mathbf{H}(\mathcal S).
\end{align*}
This is exactly the condition in (\ref{infeasibility condition}), which constitutes a converse region.

\subsection{Justification for Tight Thresholds} \label{sec: threshold proof}

\mic{In this section, we justify the tight threshold for exact set recovery established in Theorem \ref{thm: ac pair for set recovery}. The same analysis also applies to the tight threshold for exact permutation recovery in Theorem \ref{thm: ac pair for exact permutation recovery}.}

Note that the achievability condition (\ref{tight threshold achievability}) in Theorem \ref{thm: ac pair for set recovery} is precisely the condition (\ref{set achievability condition}) in Theorem \ref{thm: set recovery achievability}. Therefore, it remains to show that if any of the conditions in (\ref{tight threshold case 1}) or (\ref{tight threshold case 2}) holds, then the converse condition can be reduced to condition (\ref{tight threshold converse}).

Suppose (\ref{tight threshold case 1}) holds. The converse condition in (\ref{new converse region 1}) is simply
\begin{align*}
    \frac{m}{2}h(p) - \log n + \log m =  \frac{m}{2}h(p) - (1 - o(1))\log n \to -\infty.
\end{align*}

Suppose (\ref{tight threshold case 2}) holds. Then we are allowed to use a better bound for structural entropy in (\ref{structural entropy ER}). Therefore, the inequality in (\ref{information for decoder}) can be improved as 
\begin{align*}
    \mathbf{H}(G[\mathcal S] \mid G) \leq \mathbf{H}(G[\mathcal S]) \leq \binom{m}{2} h(p) - m! + o(1) \asymp \binom{m}{2} h(p) - m \log m.
\end{align*}

Therefore, under the simplified condition (\ref{tight threshold converse}), we still have
\begin{align*}
    \mathbf{H}(G[\mathcal S] \mid G) \leq \binom{m}{2} h(p) - m \log m \ll m \log n - m \log m \leq \mathbf{H}(\mathcal S).
\end{align*}

Hence, the converse still holds by a similar argument in Section \ref{sec: converse proof}.

\section{Conclusion and Future Research Direction} \label{sec: conclusion}

In this paper, we proposed the Erd\"os-R\'enyi subgraph pair model to study the subgraph alignment problem. We established information-theoretic limits for both exact set recovery and exact permutation recovery, including achievability and converse results. Further, our results coincide under certain conditions. 

The study of exact set recovery is not yet entirely completed. As discussed in Section \ref{app: comparison}, the achievability and converse conditions do not coincide in the most general scenario. We conjectured that a tighter converse condition can be established. 

{\color{black}
\begin{conjecture}
    Consider the Erd\"os-R\'enyi subgraph pair $(G, H_\pi) \sim \mathcal G(n,m,p)$.
    
    \noindent \underline{Achievability.} If 
    \begin{align}
        \frac{m}{2}h(p) -\log n \to \infty,
    \end{align}
    then the brute-force algorithm achieves exact set recovery w.h.p..

    \noindent \underline{Converse.} If 
    \begin{align}
        \frac{m}{2}h(p) -\log n \to -\infty,
    \end{align}
    then no algorithm achieves exact set recovery w.h.p..
\end{conjecture}
}
\mic{We believe the conditions in the above conjecture give the tight threshold for the exact subgraph alignment problem for two main reasons. The first reason is that the achievability and converse conditions coincide under certain extra assumptions, as described in Theorem \ref{thm: ac pair for set recovery} and Theorem \ref{thm: ac pair for exact permutation recovery}. The second reason is elaborated below.}

\mic{The current converse differs from the conjectured converse by an additional $\log m$ factor. This gap arises from the crude upper bound $m!$ on the automorphism number of the subgraph, which contributes an extra term of order $\log m$ (under suitable factorization). This appears in the analysis of both the standard second-moment method and the lossless compression argument. The main difficulty in eliminating this extra factor boils down to the problem of analyzing the expected automorphism number of Erd\"os-R\'enyi graphs in the general parameter regime, which is left as a possible direction for future research.}

From an algorithmic perspective, the achievability results in this work rely on a brute-force estimator, which has a prohibitively high computational complexity for the subgraph alignment problem. Therefore, a natural question is whether there is any efficient algorithms that achieves set recovery with theoretical guarantees.

Beyond the current framework, several extensions of the subgraph alignment problem are possible. One direction is to consider a more general graph model -- such as the subsampled Erd\"os-R\'enyi subgraph pair model described in \cite{shiu25}, or random graph models beyond Erd\"os-R\'enyi. In addition, in the study of graph matching, it is often common to consider other forms of recovery criteria, classic examples are almost-exact recovery and partial recovery.


%
\section*{Acknowledgment}
This work was supported in part by the Natural Sciences and Engineering Research Council of Canada (NSERC) Discovery Grant RGPIN-2019-05448, the Aarhus Universitets Forskningsfond project number AUFF 39001, and the NordForsk Nordic University Cooperation on Edge Intelligence (Grant No. 168043).

\appendices

\section{Proof of Proposition \ref{prop: equivalence}} \label{app: equivalence}
\textcolor{black}{
In the following, we show that for each $(G, H_\pi)$, the sets $\mathcal O_{\mathrm{BF}}$ and $\mathcal O_{\mathrm{MAP}}$ are the same. Equivalently, we show that 
\begin{align*}
    \Pr(\mathcal S = S, \Sigma = \sigma \mid \mathcal G = G, \mathcal H = H_\pi) \propto \mathbf{1}\{\sigma(G[S]) = H_\pi\}.
\end{align*}
By Bayes' rule, we have
\begin{align*}
    \Pr(\mathcal S = S, \Sigma = \sigma \mid \mathcal G = G, \mathcal H = H_\pi) 
    & = \frac{\Pr(\mathcal G = G, \mathcal H = H_\pi \mid \mathcal S = S, \Sigma = \sigma)\Pr(\mathcal S = S, \Sigma = \sigma)}{\Pr(\mathcal G = G, \mathcal H = H_\pi)} \\
    & = \frac{\Pr(\mathcal G = G, \mathcal H = H_\pi \mid \mathcal S = S, \Sigma = \sigma)\Pr(\mathcal S = S, \Sigma = \sigma)}{\sum_{(S', \sigma') \in \Theta} \Pr(\mathcal G = G, \mathcal H = H_\pi \mid \mathcal S = S', \Sigma = \sigma')
    \Pr(\mathcal S = S', \Sigma = \sigma')} \\
    & = \frac{\Pr(\mathcal G = G, \mathcal H = H_\pi \mid \mathcal S = S, \Sigma = \sigma)}{\sum_{(S', \sigma') \in \Theta} \Pr(\mathcal G = G, \mathcal H = H_\pi \mid \mathcal S = S', \Sigma = \sigma')} \\
    & \propto \Pr(\mathcal G = G, \mathcal H = H_\pi \mid \mathcal S = S, \Sigma = \sigma).
\end{align*}
On the other hand, $H_\pi = (G[S])_\pi$ is a function of $G, S, \pi$, so
\begin{align*}
    \Pr(\mathcal G = G, \mathcal H = H_\pi \mid \mathcal S = S, \Sigma = \sigma) & = \Pr(\mathcal G = G \mid \mathcal S = S, \Sigma = \sigma) \Pr(\mathcal H = H_\pi \mid \mathcal G = G, \mathcal S = S, \Sigma = \sigma) \\
    & = \Pr(\mathcal G = G \mid \mathcal S = S, \Sigma = \sigma) \mathbf{1}\{\sigma(G[S]) = H_\pi\} \\
    & = \Pr(\mathcal G = G) \mathbf{1}\{\sigma(G[S]) = H_\pi\} \\
    & \propto \mathbf{1}\{\sigma(G[S]) = H_\pi\},
\end{align*}
which is the desired result. Therefore, they always output the same estimates, assuming we use the same tie-breaking rule for both algorithms.
}
\section{Comparison} \label{app: comparison}

In this section, we analyze our main results. In particular, we compare Theorem \ref{thm: set recovery achievability} and Theorem \ref{thm: converse} with the best-known results in Section \ref{sec: comparison achievability} and Section \ref{sec: comparison converse} respectively. Then, we examine the necessity of including multiple conditions in Theorem \ref{thm: perm recovery achievability}, showing that each captures distinct regimes and that neither condition is redundant. 

\mic{In our earlier \lele{conference version}~\cite{shiu25}, we established information-theoretic results, which act as a baseline for comparison to this paper. \lele{Let us first recall the achievability and converse regions as follows.}}

\begin{theorem}[\mic{Information-theoretic limit in \cite{shiu25}}] \label{thm: best-known}
    Consider the Erd\"os-R\'enyi subgraph pair $(G, H_\pi) \sim \mathcal G(n,m,p)$. 

    \noindent\underline{Achievability:} If 
    \begin{align} \label{best known achievability}
        \frac{m}{2}\log\frac{1}{1-p} - \log n \to \infty,
    \end{align}
    then Algorithm \ref{alg: brute-force algo} achieves exact set recovery w.h.p.

    \noindent\underline{Converse:} If
    
    \begin{align} 
        \begin{cases}
            m^2p = O(1), \\
            m = O (\sqrt{n}), 
        \end{cases}  \label{old converse region 1}
    \end{align}
    then no algorithm can achieve exact set recovery w.h.p..
\end{theorem}

\subsection{Comparison between Achievability Results} \label{sec: comparison achievability}
{\color{black}
In this section, we show that our new achievability region \eqref{set achievability condition} in Theorem \ref{thm: set recovery achievability} is an improvement of the achievability region \eqref{best known achievability} in Theorem \ref{thm: best-known}. In addition, we highlight that the improvement is strict for most $p$.\footnote{A strict improvement refers to an improvement in the asymptotic region.} We now formally state the improvement on the achievability region in the following lemma.

\begin{lemma}\label{lemma: achievability improvement}
    Under the condition $\frac{m}{2}\log(1-p) - \log n \to \infty$ in (\ref{best known achievability}), the difference between the \lele{left-hand-side (LHS)} of conditions (\ref{set achievability condition}) and (\ref{best known achievability})
    \begin{align*}
        \left(\frac{m}{2}\log(1-p)- \log n\right) - \left(\frac{m}{2}h(p) - \log n\right),
    \end{align*}
    is always non-negative for any $p \in [0,1/2]$ and $m \in \mathbb{N}$. Moreover, the difference is unbounded as $n \to \infty$ if any of the following conditions are satisfied:
    \begin{align}
        \mathrm{(i)} \quad & p(n) \not\to \frac{1}{2}, \qquad \text{ or } \\
        \mathrm{(ii)} \quad & p(n) \to \frac{1}{2} \quad \text{ and } \quad p = \frac{1}{2} - \omega(m^{-1}).
    \end{align}
\end{lemma}

\begin{proof}
    The difference between the LHS of conditions (\ref{set achievability condition}) and (\ref{best known achievability}) is 
    \begin{align} \label{achievability gap}
        \left(\frac{m}{2}h(p) - \log n\right) - \left(\frac{m}{2}\log\frac{1}{1-p} - \log n \right) & = \frac{m}{2} \left(h(p) - \log\frac{1}{1-p}\right) = \frac{m}{2}p\log \frac{1-p}{p},
    \end{align}
    which is clearly positive for all $p \in [0,1/2]$ and $m \in \mathbb{N}$.

    Note that condition (\ref{best known achievability}) implicitly implies that $m\to \infty$ as $n \to \infty$. Observe that the function $\alpha(p) \triangleq p \log \frac{1-p}{p}$ vanishes at $0$ and $1/2$ on $[0,1/2]$, we split into three cases when (i) there are some universal constant $\psi, \phi > 0$ such that $p(n) \in [\psi, 1/2-\phi]$ for sufficiently large $n$; (ii) $p(n) \to 0$ as $n \to \infty$; (iii) $p(n) \to 1/2$ as $n \to \infty$.

    \begin{itemize}
        \item (Case i) Note that $\alpha(p)$ is continuous and concave on $p \in [\psi, 1/2 - \phi]$. For sufficiently large $n$, $p(n) \geq \min\{\alpha(\psi), \alpha(\phi)\} > 0$ is bounded away from $0$. Therefore, (\ref{achievability gap}) is lower bounded by 
        \begin{align*}
            \frac{m}{2}p \log \frac{1-p}{p} \geq \frac{m}{2} \min\{\alpha(\psi), \alpha(\phi)\} \to \infty.
        \end{align*}

        \item (Case ii) For small $p > 0$, we have $p \log \frac{1-p}{p} \geq \log \frac{1}{1-p}$. Specifically, this holds when when $p \leq p^*$, where $p^* \approx 0.2415$ is the unique solution of the equation $(1+p)\log(1-p) - p \log p = 0$. Again we use the fact that condition (\ref{best known achievability}) implicitly implies that $\frac{m}{2}\log \frac{1}{1-p} \to \infty$ as $n \to \infty$, so (\ref{achievability gap}) is lower bounded by 
        \begin{align*}
            \frac{m}{2}p \log \frac{1-p}{p} \geq \frac{m}{2} \log \frac{1}{1-p} \to \infty
        \end{align*}
        for $p \to 0$.

        \item (Case iii) Let $q(n) = 1/2 - p(n)$. Clearly $q(n) \to 0$ as $n \to \infty$. Observe that the asymptotic behavior of (\ref{achievability gap}) is
        \begin{align*}
            \frac{m}{2}p \log \frac{1-p}{p} & = \frac{m}{2}p \log \left(\frac{\frac{1}{2}+q}{\frac{1}{2}-q}\right) \\
            & = \frac{m}{2}p \log \left(1 + \frac{4\lele{q}}{1-2q}\right) \\
            & = \frac{m}{2}p\cdot (4q) (1 + o(1)) \\
            & = 2mpq (1 + o(1)) \\
            & \to mq.
        \end{align*}
        Therefore, if $p = \frac{1}{2} - \omega(m^{-1})$, or equivalently, $q = \omega(m^{-1})$,
        \begin{align*}
            \frac{m}{2}p\log \frac{p}{1-p} \to mq \to \infty.
        \end{align*}
    \end{itemize}
\end{proof}

\begin{remark}
    We point out that the condition $p = \frac{1}{2} - \omega(m^{-1})$ is necessary for a strict improvement. Consider the triplet $(n,m(n),p(n))$ given by $(n,n/2, 1/2 - 1/n)$, which satisfy (\ref{best known achievability}) as 
    \begin{align*}
        \frac{m}{2}\log\frac{1}{1-p} - \log n & = \frac{n}{4}\log \frac{1}{\frac{1}{2}+\frac{1}{n}} - \log n = \frac{n}{4}\log \left(\frac{2n}{n+2}\right) - \log n \approx \frac{\log 2}{4}n - \log n \to \infty.
    \end{align*}
    However, the gap in (\ref{achievability gap})
    \begin{align*}
        \frac{m}{2}p\log \frac{1-p}{p} & = \frac{np}{4} \log \frac{\frac{1}{2}+\frac{1}{n}}{\frac{1}{2}-\frac{1}{n}} = \frac{p}{4}\log \left(1+\frac{4}{n-2}\right)^n \to \frac{1}{8} \log e^4= \frac{1}{2}
    \end{align*}
    is bounded, which does not lead to a strict improvement.
\end{remark}

\subsection{Comparison between Converse Results} \label{sec: comparison converse}

In this section, we show that our new converse region (\ref{new converse region 1}) in Theorem \ref{thm: converse} is a \lele{strict} improvement of the converse region (\ref{old converse region 1}) of the converse region (\ref{old converse region 1}) in Theorem \ref{thm: best-known}. In particular, our new converse region doesn't have any restrictions on $m$. We now formally state the improvement on the converse region in the following lemma.

\begin{lemma} \label{lemma: converse improvement}
    Any triplet sequence $(n,m(n),p(n))$ that satisfies condition (\ref{old converse region 1}) also satisfies condition (\ref{new converse region 1}).
\end{lemma}

\begin{proof}
    Let $(n,m(n), p(n))$ be a triplet sequence that satisfies condition (\ref{old converse region 1}). Then we can find some universal constants $C, C' > 0$ such that 
    \begin{align*}
        m^2p \leq C \qquad \text{ and } \qquad m \leq C' \sqrt{n}
    \end{align*}
    for large $n$. We split into three cases when (i) $p \neq o(1)$; (ii) $p = o(1)$ and $m$ is bounded; (iii) $p = o(1)$ and $m$ is unbounded.

    \begin{itemize}
        \item (Case i) Since $p \neq o(1)$, we can find some universal constant $\eta > 0$ such that $p(n) \geq \eta$ for sufficiently large $n$. Therefore, $m \leq \sqrt{\frac{C}{p}} \leq \sqrt{\frac{C}{\eta}} \triangleq \tilde{C}$ is bounded. Therefore, (\ref{new converse region 1}) holds as
        \begin{align*}
            \frac{m}{2}h(p) - \log \frac{n}{m} \leq \frac{m}{2} - \log n + \log m \leq \frac{\tilde{C}}{2} + \log \tilde{C} - \log n \to -\infty.
        \end{align*}

        \item (Case ii) Note that the binary entropy function $h(p)$ can be upper bounded by 
        \begin{align} \label{binary entropy upper bound}
            h(p) = -p \log p - (1-p) \log (1-p) \leq p \log \frac{1}{p}.
        \end{align}
        Also, note that $h(p)$ is increasing on $p \in [0,1/2]$, we have
        \begin{align*}
            h(p) \leq h\left(\frac{C}{m^2}\right) \leq \frac{C}{m^2}\log \frac{m^2}{C}
        \end{align*}
        by the approximation in (\ref{binary entropy upper bound}). Now observe that 
        \begin{align}\label{inequality for converse improvement}
            \frac{m}{2}h(p) - \log \frac{n}{m} & \leq \frac{m}{2}\cdot\frac{C}{m^2}\log\frac{m^2}{C} - \log n + \log m  = \left(\frac{C}{m} + 1 \right) \log m - \frac{C}{2m} \log C- \log n.
        \end{align}

        Since $m$ is bounded, there is some universal constant $\tilde{C}'$ such that $\left(\frac{C}{m}+1\right) \log m - \frac{C}{2m}\log C \leq \tilde{C}'$ for all $n$. Therefore, (\ref{new converse region 1}) holds as
        \begin{align*}
            \frac{m}{2}h(p) - \log \frac{n}{m} \leq \tilde{C}' - \log n \to -\infty.
        \end{align*}

        \item (Case iii) The inequality in (\ref{inequality for converse improvement}) still holds for $p =o(1)$ and unbounded $m$. When $m$ is unbounded, $\frac{C}{m} = o(1)$. Since $m \leq C' \sqrt{n}$,
        \begin{align*}
            \left(\frac{C}{m} + 1 \right) \log m - \frac{C}{2m} \log C- \log n & \leq (1 + o(1)) \log C'\sqrt{n} - \frac{C}{2m} \log C - \log n \\
            & \leq \frac{1+o(1)}{2}\log n + (1+o(1)) \log C' - o(1) - \log n \\
            & = -\frac{1}{2}\log n + o(1) \\
            & \to -\infty,
        \end{align*}
        which completes the proof.
    \end{itemize}
\end{proof}
}

\section{Justification for the Achievability-Converse Pairs} \label{app: justification for ac pairs}

\mic{In this section, we first justify Theorem \ref{thm: ac pair for set recovery} and Theorem \ref{thm: ac pair for exact permutation recovery}, which state a bidirectional fundamental limit statement for exact set recovery and exact permutation recovery, respectively.}

\mic{Both the achievability and converse statements in Theorem \ref{thm: ac pair for set recovery} is simply collecting conditions from Theorem \ref{thm: set recovery achievability} and Theorem \ref{thm: converse} respectively. Also, the achievability statement in Theorem \ref{thm: ac pair for exact permutation recovery} comes from Theorem \ref{thm: perm recovery achievability}. Therefore, in the following, we focus on the converse statement in Theorem \ref{thm: ac pair for exact permutation recovery}.
}
\mic{Note that} 
\begin{align*}
    \textcolor{black}{\Pr(\mathcal S = \hat{\mathcal S}, \pi = \hat{\pi}) \leq \min\{\Pr(\mathcal S = \hat{\mathcal S}), \Pr(\pi = \hat{\pi}) \}.}
\end{align*}
\mic{If either $\Pr(\mathcal S = \hat{\mathcal S})\not \to1$ or $\Pr(\pi = \hat{\pi}) \not\to 1$, then $\Pr(\mathcal S = \hat{\mathcal S}, \pi = \hat{\pi}) \not\to 1$. Applying Theorem \ref{thm: converse} and Lemma \ref{lemma: wright}, we know}
\begin{align*} 
    \mic{\frac{m}{2}h(p) - \log \frac{n}{m} \to -\infty} & \mic{\implies \Pr(\mathcal S = \hat{\mathcal S}) \not\to 1,} \\
    \mic{mp - \log m \to -\infty} & \mic{\implies \Pr(\pi = \hat{\pi}) \not \to 1.}
\end{align*}
\mic{It follows that if any of the conditions in (\ref{ac pair converse permutation recovery}) holds, then $\Pr(\mathcal S = \hat{\mathcal S}, \pi = \hat{\pi})$ is bounded away from 1, and hence constitutes a converse for exact permutation recovery.}

\section{Proof of Lemma \ref{lemma: complement graph}} \label{app: proof of complement graph lemma}

Let $G = ([n], E_G)$. Then
\begin{itemize}
    \item $H = (\mathcal S, \{\{u,v\} \in E_G \mid u, v \in \mathcal S\})$;
    \item $\tilde{H} = (\mathcal S, \{\{u,v\} \in \binom{[n]}{2} \setminus E_G \mid u, v \in \mathcal S\})$;
    \item $H_\pi = ([m], \{\{\pi(u), \pi(v)\} \mid \{u,v\} \in E_H\})$;
    \item $H_\pi^c = ([m], \{\{\pi(u), \pi(v)\} \in \binom{[m]}{2} \setminus E_{H_\pi}\})$.
\end{itemize}

It remains to show that $\{u,v \}\in E_{\tilde{H}}$ if and only if $\{\pi(u), \pi(v)\} \in E_{H_\pi^c}$, which holds since
\begin{align*}
    \{u,v\} \in E_{\tilde{H}} & \iff \{u, v\} \notin E_H 
    \iff \{\pi(u), \pi(v)\} \notin E_{H_\pi} \iff \{\pi(u), \pi(v)\} \in E_{H_\pi^c}.
\end{align*}

\section{Proof of Lemma \ref{lemma: p <= 1/2}} \label{app: proof of p <= 1/2}

Let $N_2 = \binom{n}{2}$. Label the edges of a complete graph on $n$ vertices by $e_1, \ldots,e_{N_2}$. For each $1 \leq j \leq N_2$, let $X_j$ be the binary random variable indicating the presence of an edge $e_j$. For $G \sim \mathrm{ER}(n,p)$, the edge variables $\{X_j\}_{j=1}^{N_2}$ are i.i.d.~$\mathrm{Bern}(p)$. Then the random variables $\{1 - X_j\}_{j=1}^{N_2}$ are i.i.d.~$\mathrm{Bern}(1-p)$, which corresponds to the edge variables of the complement graph $G^c$. In other words, $G^c \sim \mathrm{ER}(n,1-p)$.

\section{Proof of Lemma \ref{lemma: close form expectation}} \label{app: proof of close form}

\begin{proof}
   Let $\mathcal S \subseteq [n]$ with $|\mathcal S| = v_H \triangleq m$. Without loss of generality, we may assume $\mathcal S = [m]$. From Definition~\ref{defn:XH}, $$\mathbb{E}(X_H) = \binom{n}{v_H} \Pr(G[m]\cong H).$$

   To compute $\Pr(G[m]\cong H)$, we want to count the number of subgraphs on the vertex set $[m]$ that are isomorphic to $H$. For each permutation $\sigma: [m] \to [m]$, let $H_\sigma$ denote the graph obtained by permuting the vertices of $H$ by $\sigma$ and an edge $\{\sigma(u),\sigma(v)\}$ is connected in $H_\sigma$ iff the edge $\{u,v\}$ is connected in $H$. Since each $H_\sigma \cong H$, this leads to exactly $v_H!$ graphs on $[m]$ that is isomorphic to $H$.
   
   However, some of them are double-counted, because the adjacency matrices of $H$ and $H_{\sigma}$ are identical (under the same vertex ordering) iff $\sigma$ is an automorphism of $H$. Therefore, the number of distinct subgraphs on $[m]$ that are isomorphic to $H$ is $\frac{v_H!}{\mathrm{Aut}_H}$. According to the Erd\"os-R\'enyi measure, each realization occurs with probability $p^{e_H}(1-p)^{\binom{v_H}{2} - e_H}$, so we have $\Pr(G[m] \cong H) = \frac{v_H!}{\mathrm{Aut}_H}p^{e_H}(1-p)^{\binom{v_H}{2} - e_H}$. The result follows from the number of choices of $\mathcal S$.
\end{proof}



\ifCLASSOPTIONcaptionsoff
  \newpage
\fi



%
\bibliographystyle{IEEEtran}
\bibliography{references}

@article{wri97,
  title={Graphs on unlabelled nodes with a given number of edges},
  author={Wright, EM},
  journal={Acta Mathematica},
  volume={126},
  number={1},
  pages={1--9},
  year={1971},
  publisher={Springer}
}

@inproceedings{lan23,
  title={Sub-gmn: The neural subgraph matching network model},
  author={Lan, Zixun and Yu, Limin and Yuan, Linglong and Wu, Zili and Niu, Qiang and Ma, Fei},
  booktitle={2023 16th International Congress on Image and Signal Processing, BioMedical Engineering and Informatics (CISP-BMEI)},
  pages={1--7},
  year={2023},
}

@incollection{sho00,
  title={Graph-theoretical methods in computer vision},
  author={Shokoufandeh, Ali and Dickinson, Sven},
  booktitle={Summer School on Theoretical Aspects of Computer Science},
  pages={148--174},
  year={2000},
  publisher={Springer}
}

@inproceedings{das20,
  title={Discovering interesting subgraphs in social media networks},
  author={Dasgupta, Subhasis and Gupta, Amarnath},
  booktitle={2020 IEEE/ACM International Conference on Advances in Social Networks Analysis and Mining (ASONAM)},
  pages={105--109},
  year={2020},
}

@article{tia07,
  title={SAGA: a subgraph matching tool for biological graphs},
  author={Tian, Yuanyuan and Mceachin, Richard C and Santos, Carlos and States, David J and Patel, Jignesh M},
  journal={Bioinformatics},
  volume={23},
  number={2},
  pages={232--239},
  year={2007},
  publisher={Oxford University Press}
}

@inproceedings{zha09,
  title={GADDI: distance index based subgraph matching in biological networks},
  author={Zhang, Shijie and Li, Shirong and Yang, Jiong},
  booktitle={Proceedings of the 12th international conference on extending database technology: advances in database technology},
  pages={192--203},
  year={2009}
}

@article{ma18,
  title={A comparative study of subgraph matching isomorphic methods in social networks},
  author={Ma, Tinghuai and Yu, Siyang and Cao, Jie and Tian, Yuan and Al-Dhelaan, Abdullah and Al-Rodhaan, Mznah},
  journal={IEEE Access},
  volume={6},
  pages={66621--66631},
  year={2018},
  publisher={IEEE}
}

@article{lla01,
  title={Symbol recognition by error-tolerant subgraph matching between region adjacency graphs},
  author={Llad{\'o}s, Josep and Mart{\'\i}, Enric and Villanueva, Juan J.},
  journal={IEEE Transactions on Pattern Analysis and Machine Intelligence},
  volume={23},
  number={10},
  pages={1137--1143},
  year={2001},
  publisher={IEEE}
}

@incollection{sin08,
  title={Global alignment of multiple protein interaction networks},
  author={Singh, Rohit and Xu, Jinbo and Berger, Bonnie},
  booktitle={Biocomputing},
  pages={303--314},
  year={2008},
}

@inproceedings{hag05,
  title={Robust textual inference via graph matching},
  author={Haghighi, Aria and Ng, Andrew Y and Manning, Christopher D},
  booktitle={Proceedings of Human Language Technology Conference and Conference on Empirical Methods in Natural Language Processing},
  pages={387--394},
  year={2005}
}

@inproceedings{borgelt2002,
author = {Borgelt, Christian and Berthold, Michael R.},
title = {Mining Molecular Fragments: Finding Relevant Substructures of Molecules},
year = {2002},
booktitle = {Proceedings of the 2002 IEEE International Conference on Data Mining},
}

@inproceedings{fan2012,
author = {Fan, Wenfei},
title = {Graph pattern matching revised for social network analysis},
year = {2012},
booktitle = {Proceedings of the 15th International Conference on Database Theory},
pages = {8–21},
location = {Berlin, Germany},
}

@inproceedings{ped11,
  title={On the privacy of anonymized networks},
  author={Pedarsani, Pedram and Grossglauser, Matthias},
  booktitle={Proceedings of the 17th ACM SIGKDD international conference on Knowledge discovery and data mining},
  pages={1235--1243},
  year={2011}
}

@article{cul16,
  title={Improved achievability and converse bounds for erdos-r{\'e}nyi graph matching},
  author={Cullina, Daniel and Kiyavash, Negar},
  journal={ACM SIGMETRICS performance evaluation review},
  volume={44},
  number={1},
  pages={63--72},
  year={2016},
  publisher={ACM New York, NY, USA}
}

@article{cul17,
  title={Exact alignment recovery for correlated {E}rdos {R}enyi graphs},
  author={Cullina, Daniel and Kiyavash, Negar},
  journal={arXiv preprint arXiv:1711.06783},
  year={2017}
}

@article{wu22,
  title={Settling the sharp reconstruction thresholds of random graph matching},
  author={Wu, Yihong and Xu, Jiaming and Sophie, H Yu},
  journal={IEEE Transactions on Information Theory},
  volume={68},
  number={8},
  pages={5391--5417},
  year={2022},
  publisher={IEEE}
}

@article{dai19,
  title={Analysis of a canonical labeling algorithm for the alignment of correlated erdos-r{\'e}nyi graphs},
  author={Dai, Osman Emre and Cullina, Daniel and Kiyavash, Negar and Grossglauser, Matthias},
  journal={Proceedings of the ACM on Measurement and Analysis of Computing Systems},
  volume={3},
  number={2},
  pages={1--25},
  year={2019},
}

@inproceedings{fan20,
  title={Spectral graph matching and regularized quadratic relaxations: Algorithm and theory},
  author={Fan, Zhou and Mao, Cheng and Wu, Yihong and Xu, Jiaming},
  booktitle={International conference on machine learning},
  pages={2985--2995},
  year={2020},
}

@article{din21,
  title={Efficient random graph matching via degree profiles},
  author={Ding, Jian and Ma, Zongming and Wu, Yihong and Xu, Jiaming},
  journal={Probability Theory and Related Fields},
  volume={179},
  pages={29--115},
  year={2021},
  publisher={Springer}
}

@inproceedings{mao21,
  title={Random graph matching with improved noise robustness},
  author={Mao, Cheng and Rudelson, Mark and Tikhomirov, Konstantin},
  booktitle={Conference on Learning Theory},
  pages={3296--3329},
  year={2021},
}

@article{din23,
  title={A polynomial-time iterative algorithm for random graph matching with non-vanishing correlation},
  author={Ding, Jian and Li, Zhangsong},
  journal={arXiv preprint arXiv:2306.00266},
  year={2023}
}

@article{mao23a,
  title={Exact matching of random graphs with constant correlation},
  author={Mao, Cheng and Rudelson, Mark and Tikhomirov, Konstantin},
  journal={Probability Theory and Related Fields},
  volume={186},
  number={1},
  pages={327--389},
  year={2023},
  publisher={Springer}
}

@inproceedings{mao23b,
  title={Random graph matching at Otter’s threshold via counting chandeliers},
  author={Mao, Cheng and Wu, Yihong and Xu, Jiaming and Yu, Sophie H},
  booktitle={Proceedings of the 55th Annual ACM Symposium on Theory of Computing},
  pages={1345--1356},
  year={2023}
}

@article{kor14,
  title={An efficient reconciliation algorithm for social networks},
  author={Korula, Nitish and Lattanzi, Silvio},
  journal={Proceedings of the VLDB Endowment},
  volume={7},
  number={5},
  year={2014}
}

@inproceedings{yar13,
  title={On the performance of percolation graph matching},
  author={Yartseva, Lyudmila and Grossglauser, Matthias},
  booktitle={Proceedings of the first ACM conference on Online social networks},
  pages={119--130},
  year={2013}
}

@article{mos20,
  title={Seeded graph matching via large neighborhood statistics},
  author={Mossel, Elchanan and Xu, Jiaming},
  journal={Random Structures \& Algorithms},
  volume={57},
  number={3},
  pages={570--611},
  year={2020},
  publisher={Wiley Online Library}
}

@inproceedings{shi17,
  title={Seeded graph matching: Efficient algorithms and theoretical guarantees},
  author={Shirani, Farhad and Garg, Siddharth and Erkip, Elza},
  booktitle={2017 51st Asilomar Conference on Signals, Systems, and Computers},
  pages={253--257},
  year={2017},
}

@article{wan23,
  title={Efficient Algorithms for Attributed Graph Alignment with Vanishing Edge Correlation},
  author={Wang, Ziao and Wang, Weina and Wang, Lele},
  journal={arXiv preprint arXiv:2308.09210},
  year={2023}
}

@article{zha24,
  title={Attributed graph alignment},
  author={Zhang, Ning and Wang, Ziao and Wang, Weina and Wang, Lele},
  journal={IEEE Transactions on Information Theory},
  year={2024},
  publisher={IEEE}
}

@article{wan24,
  title={On the feasible region of efficient algorithms for attributed graph alignment},
  author={Wang, Ziao and Zhang, Ning and Wang, Weina and Wang, Lele},
  journal={IEEE Transactions on Information Theory},
  year={2024},
  publisher={IEEE}
}

@inproceedings{cul18,
  title={Fundamental limits of database alignment},
  author={Cullina, Daniel and Mittal, Prateek and Kiyavash, Negar},
  booktitle={2018 IEEE International Symposium on Information Theory (ISIT)},
  pages={651--655},
  year={2018},
}

@inproceedings{nar09,
  title={De-anonymizing social networks},
  author={Narayanan, Arvind and Shmatikov, Vitaly},
  booktitle={2009 30th IEEE symposium on security and privacy},
  pages={173--187},
  year={2009},
}

@article{pru15,
  title={Graph-based word alignment for clinical language evaluation},
  author={Prud'hommeaux, Emily and Roark, Brian},
  journal={Computational Linguistics},
  volume={41},
  number={4},
  pages={549--578},
  year={2015},
  publisher={MIT Press One Rogers Street, Cambridge, MA 02142-1209, USA journals-info~…}
}

@inproceedings{kou13,
  title={Big-align: Fast bipartite graph alignment},
  author={Koutra, Danai and Tong, Hanghang and Lubensky, David},
  booktitle={2013 IEEE 13th international conference on data mining},
  pages={389--398},
  year={2013},
  organization={IEEE}
}

@article{rau19,
  title={Bit-parallel sequence-to-graph alignment},
  author={Rautiainen, Mikko and M{\"a}kinen, Veli and Marschall, Tobias},
  journal={Bioinformatics},
  volume={35},
  number={19},
  pages={3599--3607},
  year={2019},
  publisher={Oxford University Press}
}

@article{she12,
  title={Mining functional subgraphs from cancer protein-protein interaction networks},
  author={Shen, Ru and Goonesekere, Nalin CW and Guda, Chittibabu},
  journal={BMC systems biology},
  volume={6},
  number={Suppl 3},
  pages={S2},
  year={2012},
  publisher={Springer}
}

@inproceedings{bol76,
  title={Cliques in random graphs},
  author={Bollob{\'a}s, B{\'e}la and Erd{\"o}s, Paul},
  booktitle={Mathematical Proceedings of the Cambridge Philosophical Society},
  volume={80},
  number={3},
  pages={419--427},
  year={1976},
  organization={Cambridge University Press}
}

@article{alo98,
  title={Finding a large hidden clique in a random graph},
  author={Alon, Noga and Krivelevich, Michael and Sudakov, Benny},
  journal={Random Structures \& Algorithms},
  volume={13},
  number={3-4},
  pages={457--466},
  year={1998},
  publisher={Wiley Online Library}
}

@inproceedings{mos23,
  title={Sharp thresholds in inference of planted subgraphs},
  author={Mossel, Elchanan and Niles-Weed, Jonathan and Sohn, Youngtak and Sun, Nike and Zadik, Ilias},
  booktitle={The Thirty Sixth Annual Conference on Learning Theory},
  pages={5573--5577},
  year={2023},
  organization={PMLR}
}

@article{lee25,
  title={The Fundamental Limits of Recovering Planted Subgraphs},
  author={Lee, Daniel and Pernice, Francisco and Rajaraman, Amit and Zadik, Ilias},
  journal={arXiv preprint arXiv:2503.15723},
  year={2025}
}

@article{cho12,
  title={Compression of graphical structures: Fundamental limits, algorithms, and experiments},
  author={Choi, Yongwook and Szpankowski, Wojciech},
  journal={IEEE Transactions on Information Theory},
  volume={58},
  number={2},
  pages={620--638},
  year={2012},
  publisher={IEEE}
}

@article{har82,
  title={Computers and intractability: a guide to the theory of np-completeness (michael r. garey and david s. johnson)},
  author={Hartmanis, Juris},
  journal={Siam Review},
  volume={24},
  number={1},
  pages={90},
  year={1982},
  publisher={Society for Industrial and Applied Mathematics}
}

@article{sha08,
  title={Taming verification hardness: an efficient algorithm for testing subgraph isomorphism},
  author={Shang, Haichuan and Zhang, Ying and Lin, Xuemin and Yu, Jeffrey Xu},
  journal={Proceedings of the VLDB Endowment},
  volume={1},
  number={1},
  pages={364--375},
  year={2008},
  publisher={VLDB Endowment}
}

@incollection{he10,
  title={Query language and access methods for graph databases},
  author={He, Huahai and Singh, Ambuj K},
  booktitle={Managing and mining graph data},
  pages={125--160},
  year={2010},
  publisher={Springer}
}

@inproceedings{bab16,
  title={Graph isomorphism in quasipolynomial time},
  author={Babai, L{\'a}szl{\'o}},
  booktitle={Proceedings of the forty-eighth annual ACM symposium on Theory of Computing},
  pages={684--697},
  year={2016}
}

@article{mck81,
  title={Practical graph isomorphism},
  author={MCKAY, BD},
  journal={Congressus Numeranitum},
  volume={30},
  pages={45--87},
  year={1981}
}

@article{sch76,
  title={A fast backtracking algorithm to test directed graphs for isomorphism using distance matrices},
  author={Schmidt, Douglas C and Druffel, Larry E},
  journal={Journal of the ACM (JACM)},
  volume={23},
  number={3},
  pages={433--445},
  year={1976},
  publisher={ACM New York, NY, USA}
}

@article{sto19,
  title={New exact and heuristic algorithms for graph automorphism group and graph isomorphism},
  author={Stoichev, Stoicho D},
  journal={Journal of Experimental Algorithmics (JEA)},
  volume={24},
  pages={1--27},
  year={2019},
  publisher={ACM New York, NY, USA}
}

@article{erd60,
  title={On the evolution of random graphs},
  author={Erd6s, Paul and R{\'e}nyi, Alfr{\'e}d},
  journal={Publ. Math. Inst. Hungar. Acad. Sci},
  volume={5},
  pages={17--61},
  year={1960}
}

@incollection{bol11,
  title={Random graphs},
  author={Bollob{\'a}s, B{\'e}la},
  booktitle={Modern graph theory},
  pages={215--252},
  year={2011},
  publisher={Springer}
}

@inproceedings{shiu25,
  title={On the Information-Theoretic Limit of Subgraph Alignment},
  author={Shiu, Chun Hei Michael and Cheng, Hei Victor and Wang, Lele},
  booktitle={2025 IEEE International Symposium on Information Theory (ISIT)},
  pages={1--6},
  year={2025},
  organization={IEEE}
}

@article{sha53,
  title={The lattice theory of information},
  author={Shannon, Claude},
  journal={Transactions of the IRE professional Group on Information Theory},
  volume={1},
  number={1},
  pages={105--107},
  year={1953},
  publisher={IEEE}
}

@book{aro09,
  title={Computational complexity: a modern approach},
  author={Arora, Sanjeev and Barak, Boaz},
  year={2009},
  publisher={Cambridge University Press}
}

@article{kuv95,
  title={Expected complexity of graph partitioning problems},
  author={Ku{\v{c}}era, Lud{\v{e}}k},
  journal={Discrete Applied Mathematics},
  volume={57},
  number={2-3},
  pages={193--212},
  year={1995},
  publisher={Elsevier}
}
%








\end{document}